\documentclass[journal]{IEEEtran}

\usepackage{cite}
\usepackage{textcomp}
\usepackage{color}
\usepackage{amsmath,amssymb,amsfonts}
\usepackage[pdftex]{graphicx}
\usepackage[caption=false,font=footnotesize]{subfig}
   
\newcommand\figref[1]{Fig.~\ref{#1}}              
\usepackage{algpseudocode}
\usepackage{algorithm}
\usepackage{array}
\usepackage{amsthm}
\usepackage{enumitem}   
\usepackage{booktabs}
\usepackage{tikz}
\usetikzlibrary{positioning, decorations.pathreplacing}
\usepackage{xcolor}
\definecolor{darkgreen}{RGB}{0,100,0}    
\definecolor{darkblue}{RGB}{0,0,139}    
\definecolor{darkred}{rgb}{0.55, 0.0, 0.0}
\usepackage[colorlinks=true, urlcolor = darkblue, citecolor = darkgreen, linkcolor=darkred]{hyperref} 

\newcommand{\q}{\ensuremath{\boldsymbol q}}
\renewcommand{\v}{\ensuremath{\boldsymbol v}}

\renewcommand{\u}{\ensuremath{\boldsymbol u}}
\newcommand{\W}{\ensuremath{\boldsymbol W}}
\newcommand{\h}{\ensuremath{\boldsymbol h}}
\renewcommand{\d}{\ensuremath{\boldsymbol d}}
\newcommand{\V}{\ensuremath{\boldsymbol V}}
\newcommand{\I}{\ensuremath{\boldsymbol I}}
\renewcommand{\a}{\ensuremath{\boldsymbol a}}
\newcommand{\x}{\ensuremath{\boldsymbol x}}
\newcommand{\p}{\ensuremath{\boldsymbol p}}
\newcommand{\A}{\ensuremath{\boldsymbol A}}

\renewcommand{\S}{\ensuremath{\boldsymbol S}}
\newcommand{\B}{\ensuremath{\boldsymbol B}}
\newcommand{\D}{\ensuremath{\boldsymbol D}}
\newcommand{\X}{\ensuremath{\boldsymbol X}}
\newcommand{\Z}{\ensuremath{\boldsymbol Z}}
\newcommand{\Xs}{\ensuremath{\mathcal X}}
\newcommand{\n}{\ensuremath{\boldsymbol n}}
\newcommand{\y}{\ensuremath{\boldsymbol y}}
\renewcommand{\H}{\ensuremath{\boldsymbol H}}
\newcommand{\s}{\ensuremath{\boldsymbol{s}}}
\newcommand{\E}{\ensuremath{\mathbb E}}
\newcommand{\C}{\ensuremath{\mathbb C}}
\newcommand{\M}{\ensuremath{\mathbf M}}
\newcommand{\R}{\ensuremath{\mathcal R}}
\newcommand{\0}{\ensuremath{\mathbf 0}}
\newcommand{\K}{\ensuremath{\mathcal K}}
\renewcommand{\L}{\ensuremath{\mathcal L}}
\newcommand{\N}{\ensuremath{\mathcal N}}
\newcommand{\CN}{\ensuremath{\mathcal {CN}}}
\newcommand{\dcsr}{\ensuremath{\mathrm {\Gamma^{\dl}}}}
\newcommand{\ucsr}{\ensuremath{\mathrm {\Gamma^{\ul}}}}
\newcommand{\Prob}[1]{\mathcal{P}_{#1}}
\newcommand{\U}{\ensuremath{\mathbf {U}}}
\newcommand{\crlb}{\ensuremath{\mathrm {\Gamma^{\est}}}}
\newcommand{\Var}{\ensuremath{\mathrm {Var}}}
\newcommand{\dl}{\ensuremath{\mathtt{DL}}}
\newcommand{\ul}{\ensuremath{\mathtt{UL}}}
\newcommand{\est}{\ensuremath{\mathtt{est}}}
\renewcommand{\n}{\ensuremath{\mathbf {n}}}
\newcommand{\Rq}{\ensuremath{\mathrm {Rq}}}
\newcommand{\order}[1]{\ensuremath{{#1}\textsuperscript{th}}}
\newcommand{\e}{\mathrm{e}}

\newtheorem{theorem}{Theorem}
\newtheorem{remark}{Remark}
\newcommand{\secref}[1]{Section \ref{#1}}

\newcommand{\Ha}[1]{\mathbf{H}_{a,#1}}

\DeclareMathOperator{\tr}{tr}                        
\DeclareMathOperator{\rank}{rank}                    

\usepackage{scalerel}
\usepackage{tikz}
\usetikzlibrary{svg.path}
\definecolor{orcidlogocol}{HTML}{A6CE39}
\tikzset{
  orcidlogo/.pic={
    \fill[orcidlogocol] svg{M256,128c0,70.7-57.3,128-128,128C57.3,256,0,198.7,0,128C0,57.3,57.3,0,128,0C198.7,0,256,57.3,256,128z};
    \fill[white] svg{M86.3,186.2H70.9V79.1h15.4v48.4V186.2z}
                 svg{M108.9,79.1h41.6c39.6,0,57,28.3,57,53.6c0,27.5-21.5,53.6-56.8,53.6h-41.8V79.1z M124.3,172.4h24.5c34.9,0,42.9-26.5,42.9-39.7c0-21.5-13.7-39.7-43.7-39.7h-23.7V172.4z}
                 svg{M88.7,56.8c0,5.5-4.5,10.1-10.1,10.1c-5.6,0-10.1-4.6-10.1-10.1c0-5.6,4.5-10.1,10.1-10.1C84.2,46.7,88.7,51.3,88.7,56.8z};
  }
}
\newcommand\orcidicon[1]{\href{https://orcid.org/#1}{\mbox{\scalerel*{
\begin{tikzpicture}[yscale=-1,transform shape]
\pic{orcidlogo};
\end{tikzpicture}
}{|}}}}

\makeatletter
\def\ps@IEEEtitlepagestyle{%
  \def\@oddfoot{\mycopyrightnotice}%
  \def\@oddhead{\hbox{}\@IEEEheaderstyle\leftmark\hfil\thepage}\relax
  \def\@evenhead{\@IEEEheaderstyle\thepage\hfil\leftmark\hbox{}}\relax
  \def\@evenfoot{}%
}
\def\mycopyrightnotice{%
  \begin{minipage}{\textwidth}
  \centering \scriptsize
  \copyright~2025 IEEE. Personal use of this material is permitted.  Permission from IEEE must be obtained for all other uses, in any current or future media, including reprinting/republishing this material for advertising or promotional purposes, creating new collective works, for resale or redistribution to servers or lists, or reuse of any copyrighted component of this work in other works.
  \end{minipage}
}
\makeatother

\begin{document}

\bstctlcite{IEEEexample:BSTcontrol}

\title{Securing Integrated Sensing and Communication Against a Mobile Adversary: A Stackelberg Game with Deep Reinforcement Learning}

\author{Milad Tatar Mamaghani, Xiangyun Zhou,~\emph{Fellow,~IEEE}, Nan Yang,~\emph{Senior Member,~IEEE}, and \\A. Lee Swindlehurst,~\emph{Fellow,~IEEE}
\thanks{This work was supported by the Australian Research Council’s Discovery Projects funding scheme (Grant ID: DP220101318).}
\thanks{M. Tatar Mamaghani was with the Australian National University, Canberra, ACT 2601, Australia (\href{mailto:milad.tatarmamaghani@gmail.com}{\textcolor{black}{milad.tatarmamaghani@gmail.com}})}
\thanks{X. Zhou and N. Yang are with the School of Engineering, Australian National University, Canberra, ACT 2601, Australia (email: \href{mailto:xiangyun.zhou@anu.edu.au}{\textcolor{black}{xiangyun.zhou@anu.edu.au}}; \href{mailto:nan.yang@anu.edu.au}{\textcolor{black}{nan.yang@anu.edu.au}}).}
\thanks{A. L. Swindlehurst is with the Center for Pervasive Communications and
Computing, University of California, Irvine, CA 92697, USA (email: \href{mailto:swindle@uci.edu}{\textcolor{black}{swindle@uci.edu}}).}}

\maketitle

\begin{abstract}
In this paper, we study a secure integrated sensing and communication (ISAC) system employing a full-duplex base station with sensing capabilities against a mobile proactive adversarial target---a malicious unmanned aerial vehicle (M-UAV). We develop a game-theoretic model to enhance communication security, radar sensing accuracy, and power efficiency. The interaction between the legitimate network and the mobile adversary is formulated as a non-cooperative Stackelberg game (NSG), where the M-UAV acts as the leader and strategically adjusts its trajectory to improve its eavesdropping ability while conserving power and avoiding obstacles. In response, the legitimate network, acting as the follower, dynamically allocates resources to minimize network power usage while ensuring required secrecy rates and sensing performance. To address this challenging problem, we propose a low-complexity successive convex approximation (SCA) method for network resource optimization combined with a deep reinforcement learning (DRL) algorithm for adaptive M-UAV trajectory planning through sequential interactions and learning. Simulation results demonstrate the efficacy of the proposed method in addressing security challenges of dynamic ISAC systems in 6G, i.e., achieving a Stackelberg equilibrium with robust performance while mitigating the adversary’s ability to intercept network signals. 

\end{abstract}


\IEEEpeerreviewmaketitle

\section{Introduction}

\IEEEPARstart{I}{ntegrated} sensing and communication (ISAC) is an emerging technology poised to play a critical role in the success of sixth-generation (6G) wireless networks. ISAC has recently garnered significant attention from both the academic community \cite{liu2022integrated} and industry \cite{bayesteh2022integrated}. This paradigm shift moves beyond traditional separate communication and sensing systems, leveraging the shared wireless medium to perform both tasks simultaneously. ISAC not only boosts network efficiency but also unlocks new avenues for advanced applications, allowing wireless networks to function as sensors for services such as localization, gesture recognition, and object tracking, among others \cite{Wei2022}.
 
Interference management is one of the key challenges of ISAC systems.  Integrating sensing and communication functionalities into a single hardware and software unit leads to interference between the sensing and communication tasks, which must be appropriately mitigated. When properly designed, ISAC offers significant advantages, including improved spectral efficiency and reduced hardware complexity for diverse applications \cite{Lu2023, Zhang2022a}. Nevertheless, the inherent coupling between communication and sensing in ISAC introduces unique security challenges \cite{Qu2023}. Unlike conventional communication systems, the information carried by ISAC signals can leak into the sensing operation, potentially revealing sensitive data to unintended receivers. Additionally, eavesdroppers can exploit the sensing information to gain insights into the transmitted communication messages. This issue becomes more pronounced when the target itself is an eavesdropper, aiming to intercept the information embedded in the transmit waveform \cite{Meng2023}. For radar purposes, the ISAC transmitter typically focuses energy toward the target using transmit beamforming to enhance estimation accuracy. However, this exposes the waveform to the malicious target, making it easier to intercept the information. Thus, the appropriate design of ISAC systems is crucial to secure sensitive information exchange while ensuring reliable target tracking.

\subsection{Prior Works}
In the last decade, various studies have been conducted on ISAC from different perspectives, encompassing fundamental analysis \cite{Gunlu2022}, waveform design \cite{zou2024energy}, resource allocation, joint beamforming, and capacity characterization \cite{he2023full}. Amid these advancements, the security challenges inherent in ISAC systems have garnered increasing attention. Recent work has shifted focus towards the modeling and design of secure ISAC systems through the lens of physical-layer security (PLS)  \cite{mukherjee2014principles, zhou2013physical, su2020secure, bazzi2023secure, Ren2023, Hou2023, mamaghani2024performance}.
For example,  \cite{su2020secure} investigated the PLS of a multiple-input multiple-output (MIMO) ISAC system, employing artificial noise (AN) to ensure transmission secrecy while tracking radar targets and considering scenarios with target location uncertainty and erroneous channel state information (CSI). Furthermore, \cite{bazzi2023secure} studied a full-duplex (FD) multi-eavesdropper ISAC scenario where a unified waveform was designed for power-efficient sensing and secure communication to achieve a desired level of secrecy under constraints on the integrated side lobe levels for sensing. A secure ISAC system was examined in \cite{Ren2023}, where a multi-antenna base station (BS) transmits confidential messages to a single-antenna downlink (DL) user while simultaneously sensing potential eavesdroppers. 
In \cite{Hou2023}, the authors optimized transmit beamforming in a secure ISAC framework with randomly located malicious targets.

Another avenue of research involves leveraging technologies such as UAVs for enhancing the security of ISAC systems \cite{Wu2023a, Zhao2022, Wang2022}. For example, 
\cite{Wu2023a} studied real-time UAV trajectory optimization for secure ISAC, where the UAV functions as both a legitimate transmitter and a radar receiver.  
The authors in \cite{Zhao2022} tackled the challenge of maximizing the secrecy energy efficiency in a UAV-assisted ISAC system, employing iterative algorithms for joint sensing, transmission time, and UAV location optimization.  
While UAV technology offers significant benefits, it could also pose threats when utilized for malicious purposes. This duality underscores the significant challenges faced in ensuring the security of ISAC systems, particularly when mobile adversaries are also targets that are being tracked by the system. Consequently, some research initiatives have emerged to address this important issue. In \cite{Wang2022}, the authors studied a sensing-assisted uplink (UL) communication framework between a single-antenna user and an FD-BS to combat an aerial eavesdropper with a predetermined trajectory. Using radar signals for target localization and jamming, the BS ensured secure UL transmission while maximizing received signal quality. Expanding upon \cite{Wang2022}, the authors in \cite{Liu2023b} considered multiple UL users against a mobile eavesdropper with an unknown trajectory, where radar signals from the BS are employed to track the aerial eavesdropper and estimate real-time CSI via an extended Kalman filter.

\textcolor{black}{
While some prior research in secure ISAC settings has considered mobile aerial targets (e.g.,  \cite{Wang2022, Liu2023b}), these studies assume passive aerial adversaries with either predetermined or unknown trajectories. In contrast, we focus on a critical and underexplored threat model: securing ISAC systems against \textit{mobile, proactive adversarial UAV targets}, referred to as M-UAVs, which dynamically adapt their trajectories in real time to maximize eavesdropping performance while avoiding detection. These agile M-UAVs exploit favorable air-ground (AG) channel conditions and react strategically to the legitimate network's defensive actions, which creates complex bidirectional interactions that fundamentally differ from prior studies with passive threats. Addressing this challenge necessitates a solution framework grounded in sequential decision-making and hierarchical game theory to capture the evolving strategic interplay between the legitimate system and the adversary.}

To this end, we adopt a game-theoretic perspective. Game theory \cite{han2019game} is a powerful mathematical tool for analyzing strategic interactions among rational decision-makers, or \textit{players}. By examining how one player’s decisions influence others, game theory can predict outcomes and inform strategic choices. In particular, non-cooperative game theory plays a crucial role in competitive dynamics, where players are not compelled to engage in cooperative behavior. A Stackelberg game is a sequential game where players are divided into leaders and followers. Leaders make the first move to optimize their utility function, and followers respond accordingly, intending to maximize their own. In the context of secure ISAC systems, a non-cooperative Stackelberg game (NSG) framework is an effective means of modeling the interactions between the competing parties, i.e., the legitimate system and mobile proactive target. While game theory has been applied in traditional communication systems to model adversarial behavior, to the best of our knowledge it has not been proposed for the problem of defending against agile mobile threats in ISAC systems. This gap motivates the current work.

\subsection{Contributions}
In this paper, we focus on an FD-ISAC system that wishes to improve its security against a mobile proactive adversary through an NSG. Our contributions are as follows:

\begin{itemize}
    \item  \textbf{Secure ISAC model with proactive mobile target}: We model a secure FD-ISAC system that accommodates multiple legitimate DL and UL communication users while contending with a mobile proactive aerial target. The mobile adversary significantly increases the security risks of the network due to the ability of the M-UAV to find favorable AG channels that improve its ability to eavesdrop on both DL and UL communications. 
    To enhance ISAC security, we leverage AN to not only facilitate the creation of efficient radar beams for improved sensing performance, but also degrade the M-UAV's reception quality, thereby enhancing overall system confidentiality.

    \item   \textbf{Game formulation for secure ISAC system}: We adopt a game-theoretic approach to model the sequential interaction between the legitimate network and the mobile M-UAV, facilitating robust secure transmission. Specifically, our NSG formulation models the M-UAV as the leader, optimizing its trajectory over the given mission duration to maximize its utility. This involves balancing power conservation with mission objectives such as minimizing the risk of detection and enhancing the network’s achievable secrecy rate.
    The legitimate network, acting as the follower to the M-UAV’s proactive movements, responds by dynamically adjusting its resource allocation to minimize network power consumption while satisfying the required sensing performance and maintaining both UL and DL secrecy constraints. 

    \item \textbf{SCA-DRL-empowered efficient solution}:
    To devise the optimal strategies for both the legitimate network and the M-UAV, we employ a Stackelberg equilibrium framework. This involves iteratively optimizing the strategies of both players to address the challenge of achieving a global optimum in this complex hierarchical game, ultimately ensuring steady-state convergence where the leader's optimal strategy is anticipated and countered by the follower. Specifically, we employ a novel combination of successive convex optimization (SCA) \cite{Boyd2006} and deep reinforcement learning (DRL) \cite{Sutton2018} to obtain optimal strategies for both players, ensuring efficient and robust system performance. In particular, SCA offers a low-complexity solution to finding the best power allocation by optimally designing the DL beamforming, UL transmissions, and sensing beamforming, while DRL efficiently tackles the sequential nature of the game in a dynamic environment. \textcolor{black}{   Beyond algorithmic design, we conduct a rigorous convergence and complexity analysis, showing that our proposed SCA–DRL system converges to a local Stackelberg equilibrium under standard
assumptions, based on a two-time-scale framework.  Simulation results further confirm that our approach not only converges reliably but also outperforms the relative benchmark schemes,  validating its suitability for deployment in secure ISAC environments.}

\end{itemize}

\subsection{Organization}
The remainder of this paper is structured as follows: \secref{sec1_sysModel} presents the FD-ISAC system model including a mobile proactive target, followed by the adopted performance indicators. In \secref{sec2_problem}, we formulate a game-theoretic optimization framework tailored to the considered system. Our proposed SCA-DRL solution is presented in \secref{sec3_solution}. Simulation results and technical discussions are given in \secref{sec4_simulations}, and conclusions are drawn in \secref{sec5_conclusions}.

\textit{Notation}:
Throughout the paper, the following notation is employed. The symbol $\mathbb{C}^{p \times q}$ denotes the space of $p$-by-$q$ complex-valued matrices and $\mathbb{R}$ indicates the set of real numbers. Calligraphic letters such as $\L$ represents a set whose cardinality is denoted as $|\mathcal{L}|$. Upper boldface letter $\A$ denotes a matrix, lower boldface letter $\a$ denotes a vector, and $\tr(\cdot)$, $\det(\cdot)$, $(\cdot)^\top$ and $(\cdot)^\dagger$ denote the trace, determinant, transpose, and complex conjugate transpose operations, respectively. The matrix $\I_N$ and vector $\mathbf{1}_N$ denote the identity matrix and the all-ones row vector of size $N$, respectively, $\|\x\|$ denotes the Euclidean norm of a complex-valued vector $\x$, the first partial derivative of $\x(t)$ with respect to (w.r.t.) $t$ is expressed as $\dot{\x}\triangleq \frac{\partial \x(t)}{\partial t}$, $\Re\{\cdot\}$ indicates the real part of an expression, $\otimes$ is the Kronecker product, and $\A\succeq0$ indicates that $\A$ is Hermitian-positive semidefinite (PSD). A circularly symmetric complex Gaussian multivariate distribution with mean vector $\pmb{\phi}$  and covariance matrix  $\mathbf{C}$ is denoted by $\CN(\pmb{\phi}, \mathbf{C})$, and $\E(Y)$ and $\mathrm{Var}(Y)$ are the expectation and variance of random variable $Y$. Finally, the operators $\sim$, $\in$, $\forall$, $\triangleq$, and $\setminus$ read as \textit{distributed as}, \textit{is a member of}, \textit{for all}, \textit{defined as}, and \textit{set minus}, respectively.

\section{System Model}\label{sec1_sysModel}

\begin{figure}[!t]
\centering
\includegraphics[width=\linewidth]{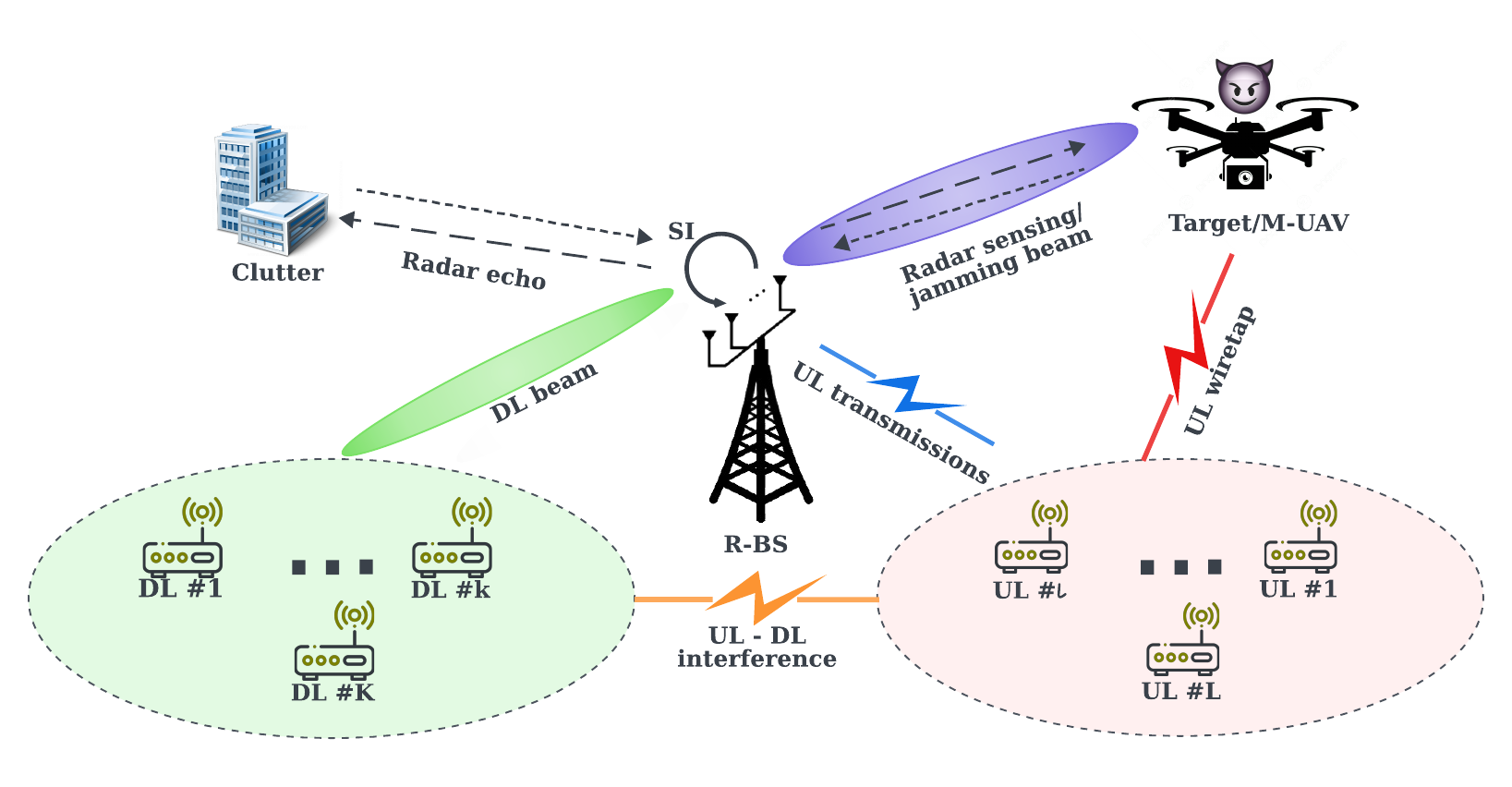}
\caption{Illustration of the ISAC system under consideration with an M-UAV.}
\label{fig_sim}
\end{figure}

We consider an ISAC system as exemplified in Fig. \ref{fig_sim}, where a dual functional multi-antenna FD-BS with radar sensing capabilities over the shared spectrum, also referred to here as a Radar-BS (R-BS),  aims to not only securely transmit confidential messages to $K$  single-antenna DL users, indexed by $\K=\{1, 2, \cdots, K\}$, but also simultaneously sense and track a target while receiving confidential uplink signals from a set of $L$  UL users, indexed by $\L=\{1, 2, \cdots, L\}$. Transmissions occur under the surveillance of an M-UAV equipped with a single receive antenna, which exploits its aerial mobility to intercept communications.

We assume that the R-BS is equipped with a uniform planar array (UPA) with $M_t= M^{Tx}_h \times M^{Tx}_v$ transmit antennas, where $M^{Tx}_h$ and $M^{Tx}_v$ indicate the number of antennas in the horizontal and vertical directions. The R-BS also has an $M_r = M^{Rx}_h \times M^{Rx}_v$ receive UPA for reception of confidential UL and radar echo signals, and thus operates in FD mode. The antennas on all devices are assumed to have omnidirectional radiation patterns. 
The R-BS, $K$ DL users, and  $L$ UL users are located at $\q_a=[x_a, y_a, 0]^\top$,  $\q_k=[x_k, y_k, 0]^\top, k\in\K$, and $\q_l=[x_l, y_l, 0]^\top, l\in\L$, respectively. 
Furthermore, we assume that the maximum mission duration of the M-UAV is \( T \). The M-UAV's trajectory is discretized into \( (N+1) \) waypoints, with \( N \) timeslots, indexed by \( \N = \{1, 2, \cdots, N\} \). Each timeslot has a duration of \( \delta_t \), satisfying \( N \delta_t \leq T \). We assume that the M-UAV can maneuver in three dimensions (3D) and express its time-varying 3D position at timeslot $n$ as $\q_e[n]=[x[n], y[n], z[n]]^\top$, with velocity $\v[n]=[v_x[n], v_y[n], v_z[n]]^\top$. We define $\v_{xy}[n]=[v_x[n],v_y[n]]^\top$ as the horizontal projection of the M-UAV's velocity at timeslot $n$. Although we assume that the location of the M-UAV is approximately fixed within each timeslot, it may vary from one timeslot to another as the M-UAV can adaptively move to improve wiretap performance. Here, we assume that the M-UAV is aware of the locations of the UL and DL users as well as the R-BS.

\subsection{Flight Power Consumption and Mobility Constraints}\label{sec:flightPower}
The power consumption of a rotary-wing UAV at a given time $n$, denoted by $P_f[n]$, can be quantified, according to classical aircraft dynamics theory as \cite{Mamaghani2021b}
\begin{align}\label{flightpow}
    P_{f}[n] &= {P_0\left(1+\frac{3\|\v_{xy}[n]\|^2}{U^2_{tip}}\right)} +{C_0 \|\v_{xy}[n]\|^3} + G_0 \|v_z[n]\|\nonumber\\
    &+ {P_1\left(\sqrt{1+\frac{\|\v_{xy}[n]\|^4}{4\nu^4_0}} - \frac{\|\v_{xy}[n]\|^2}{2\nu^2_0}\right)^{\frac{1}{2}}} , \quad\forall n
\end{align}
where $P_0$, $P_1$, and $\nu_0$ are constants representing respectively the \textit{blade profile power}, \textit{induced power}, and \textit{mean rotor-induced velocity} of the M-UAV in hovering mode, respectively. In addition,  $U_{tip}$ and $C_0$ are parameters related to aerodynamics, and $G_0$ is the weight of the M-UAV. 
The M-UAV must carefully control its movement in order to balance the requirement for high-performance signal interception with the need for manageable power consumption that extends its operational lifespan.

In our work, we impose the following constraints on the M-UAV's mobility \cite{mamaghani2023secure,TatarMamaghani2022}:
\begin{subequations}
\begin{align}
\mathrm{C1}:&~~~\q_e[1]=\q_i,~\q_e[N+1]\in  \mathbf{R}_f, N\delta_t \leq T\\
\mathrm{C2}:&~~~z^{\min}\leq z[n] \leq z^{\max},~\forall n\in\N\label{mobility_cst5}\\
\mathrm{C3}:&~~~\q_e[n] \in \R_o, \|\q_e[n] -  \{\q_o\}\| \geq d^{\mathrm{min}},~\forall n\in\N\label{mobility_cst6}\\
\mathrm{C4}:&~~~\q_e[n+1] = \v[n]\delta_t + \q_e[n],~\forall n\in\N\\
\mathrm{C5}:&~~~|x[n+1]-x[n]|\leq v^{\max}_{x}\delta_t,~\forall n\in\N\\
\mathrm{C6}:&~~~|y[n+1]-y[n]|\leq v^{\max}_{y}\delta_t,~\forall n\in\N\\
\mathrm{C7}:&~~~|z[n+1] - z[n]| \leq v^{{max}}_{z}\delta_t,~\forall n\in\N
\end{align}
\end{subequations}
where $\q_i \in \mathbb{R}^{3\times 1}$ represents the initial location of the M-UAV, $\mathbf{R}_f$ indicates the predetermined final location region where the M-UAV can land,  $z^{\min}$ and $z^{\max}$ are the minimum and maximum operational altitudes of the M-UAV, respectively, $\R_o$ indicates the M-UAV's permitted operating region, $\{\q_o\}$ indicates the location of obstacles in the region, $d^{\mathrm{min}}$ is the minimum distance that must be maintained from any obstacle to avoid collision, and $v^{{max}}_{x}$, $v^{{max}}_{y}$, and $v^{{max}}_{z}$ respectively indicate the maximum velocity of the M-UAV in the $(x,y,z)$ directions, assuming a constant speed per timeslot. 

\subsection{Full-Duplex Secure ISAC}

In the following, the time index $n$ from the signal representations is often dropped for conciseness.

\subsubsection{Communication channel models}

All communication channels are assumed to be quasi-static block fading such that the channel coefficients remain fixed in each timeslot and vary from one to another. 
We assume that the communication channels consist of two components, i.e., a large-scale fading component and a small-scale Rician fading component. We denote the base-band equivalent channels as
\begin{itemize}
    \item \(\h_{a,k} \in \mathbb{C}^{M_t \times 1}\): From the R-BS to the \order{k} DL user;
    \item \(\h_{\ell,a} \in \mathbb{C}^{M_r \times 1}\): From the \order{\ell} UL user to the R-BS;
    \item \(h_{\ell,e} \in \mathbb{C}\): From the \order{\ell} UL user to the M-UAV;
    \item \(h_{l,k} \in \mathbb{C}\): From the \order{\ell} UL user to the \order{k} DL user;
    \item \(\h_{e,a} \in \mathbb{C}^{M_r \times 1}\): From the M-UAV to the R-BS.
    \end{itemize}
Each entry of the vectors $\h_{a,k}$, $\h_{\ell, a}$, and $\h_{e, a}$ is denoted by the scalar $h_{i,j}$. The distance-dependent large-scale fading can be represented as
\begin{align}
    \eta_{i, j} = \begin{cases}
        \beta_0d^{-\alpha}_{i,j}, & \text{if LoS}\\
         \varsigma\beta_0d^{-\alpha}_{i,j}, & \text{if NLoS}
    \end{cases}
\end{align}
where $d_{i,j}=\|\q_{i}-\q_j\|$ is the distance between points $i$ and $j$, $\beta_0=({\lambda}/{4\pi})^2$ denotes the channel power gain at the reference distance of 1$\mathrm{m}$ in line-of-sight (LoS) conditions with wavelength $\lambda={C}/{f_c}$, where $C$ and $f_c$ denote the speed of light and the carrier frequency, respectively, $0 < \varsigma <1$ is the excess attenuation due to non-LoS (NLoS) propagation, and $\alpha$ is the path-loss exponent, which is typically set as 2 for AG channels. Note that the probability of existence of an LoS path is generally modeled as a function of elevation angle $\theta_{i,j}$ as follows:
\begin{align}
    P^{\mathrm{LoS}}_{i,j} = \frac{1}{1 + E_1 \exp\left(-E_2(\theta_{i,j}-E_1)\right)}, 
\end{align}
where $\theta_{i,j}=\frac{180}{\pi}\arcsin(\frac{z_i - z_j}{d_{i,j}})$, and $E_1$ and $E_2$ are constants depending on
the propagation environment such as rural, urban, or dense urban \cite{khuwaja2018survey}. Obviously, $P^{\mathrm{NLoS}}_{i,j} = 1 - P^{\mathrm{LoS}}_{i,j}$.

The small-scale Rician fading channel model consists of both LoS and NLoS components, expressed as
\begin{align}
g_{i, j}=\sqrt{\frac{\kappa_{i, j}}{\kappa_{i, j}+1}} g_{i, j}^{\mathrm{LoS}}+\sqrt{\frac{1}{\kappa_{i, j}+1}} g_{i, j}^{\mathrm{NLoS}},
\end{align}
where  $\kappa_{i, j}$ is the Rician K-factor and $g_{i,j}^{\mathrm{NLoS}} \sim\CN(0,1)$. Accordingly, each entry of the communication channels $h_{i,j}$ can be expressed as $h_{i,j} = \sqrt{\eta_{i,j}} g_{i,j}$, with expected channel power gain $\E(|h_{i,j}|^2)= \hat{P}^{\mathrm{LoS}}_{i,j} \beta_0d^{-\alpha}_{i,j}$, where $\hat{P}^{\mathrm{LoS}}_{i,j} \triangleq \varsigma+ (1-\varsigma)P^{\mathrm{LoS}}_{i,j}$ is the regularized LoS probability.

\subsubsection{Radar channel models}

The signals transmitted by the R-BS are reflected from the desired target (i.e., the M-UAV) and \textit{clutter}\footnote{Note that clutter refers to unwanted echoes that degrade the sensing performance. For instance, reflections from the environment, such as those from terrain, buildings, or precipitation, can mask weaker returns from target objects of interest.  In ISAC systems, interference from the radar component of the transmit signal also degrades the communications performance, necessitating effective interference management.}, and then received by the R-BS. We assume that there exists a dominant LoS component\footnote{\textcolor{black}{The communication and sensing channel models adopted in this work incorporate key practical aspects such as probabilistic LoS fading, radar clutter, and residual self-interference. However, certain environmental uncertainties, such as Doppler shifts, atmospheric attenuation, UAV turbulence, hardware impairments, and real-time CSI estimation errors, are not explicitly considered. These effects may introduce degradation in beamforming alignment and secrecy performance under highly dynamic conditions. Robust extensions incorporating bounded CSI error models or outage-based formulations are left for future investigation.}} between the R-BS and M-UAV (i.e., a Rician channel with $\kappa_{a, e}\rightarrow \infty$ and $P^{\mathrm{LoS}}_{a,e}\rightarrow 1$), which is a reasonable assumption for localization and tracking applications involving UAVs \cite{lyu2022joint, bazzi2023secure}. \textcolor{black}{In fact, accurate radar sensing fundamentally relies on a strong direct-path echo (i.e., LoS link). In contrast to communication systems, where multipaths are typically beneficial due to spatial diversity, NLoS contributions are usually detrimental to localization accuracy and sensing performance \cite{wei2024integrated}, and thus can be considered as part of the clutter.} 

Accordingly, the sensing channel between the R-BS and M-UAV is expressed as
$\hat{\A_0}= \zeta_0 \A_0$, where $\zeta_j$ models the  round-trip radar signal attenuation \cite{skolnik1980introduction}, and $\A_0 = \a_r \a^\dagger_t \in \C^{M_r\times M_t}$ with $\a_t\in\C^{M^{Tx}_hM^{Tx}_v\times 1}$ and $\a_r\in\C^{M^{Rx}_hM^{Rx}_v\times 1}$ representing the transmit and receive steering vectors, respectively. The transmit steering vector $\a_t$ is given by
\begin{align}
    &\a_t (\theta, \phi) = [1, \e^{\frac{j2\pi\delta}{\lambda}\sin\theta \cos\phi}, \cdots, \e^{\frac{j2\pi\delta (M^{Tx}_h-1)}{\lambda}\sin\theta \cos\phi}]^\top \nonumber\\
    &\otimes [1, \e^{\frac{j2\pi\delta}{\lambda}\sin\theta \sin\phi}, \cdots,\e^{\frac{j2\pi\delta (M^{Tx}_v-1)}{\lambda}\sin\theta \sin\phi}]^\top,
\end{align}
where  $\theta \in
[0, \pi/2]$ and $\phi\in [0, \pi]$ denote the elevation and azimuth angles of departure (AoDs) from the R-BS to the M-UAV, $\delta$ is the spacing between adjacent antennas, and $\lambda$ is the signal wavelength. The receive steering vector $\a_r$ is defined similarly. In addition, we consider the existence of $N_c$ clutter scatterers, with channels $\hat{\A_j}$ for $j=1,...,N_c$, that are characterized similar to $\hat{\A_0}$.

\subsubsection{Secure ISAC signal model}
The R-BS transmits
\begin{align}\label{transmitsignal_x}
    \x=\V \s+\s_r,
\end{align}
where the columns of 
$\V=[\v_1, \v_2,\cdots, \v_K]\in\C^{M_t\times K}$ are the beamforming vectors  for the DL users, $\s\in\C^{K\times1}$ indicates the information-bearing DL signal vector whose entries are assumed to satisfy $\E\left(\s \s^\dagger\right)=\I_K$, and $\s_r\in\C^{M_t\times1}$  is a dedicated radar signal, independent of $\s$, with
covariance matrix $\W = \E(\s_r\s^\dagger_r)$. Note that $\s_r$ has a dual role, serving not only as AN for jamming the M-UAV but also augmenting the degrees-of-freedom (DoFs) of the transmit signal, thereby enhancing sensing performance. 

\subsubsection{Received signal model at DL users}
The received signal at the \order{k} DL user 
is given by
\begin{align}\label{receivedsignal_yl}
    y_{k}&=\h_{a,k}^\dagger \x+\sum\limits_{\ell=1}^L h_{\ell,k} z_{\ell}+ n_{k}, \quad k\in\K
\end{align}
where $z_{\ell}\in\mathbb{C}$ indicates the UL symbol with power $\E(|z_{\ell}|^2)=p_{\ell}$ transmitted by user $\ell$, and $n_{k}$ is additive white Gaussian noise (AWGN) with variance $\sigma^2_{k}$ at the \order{k} DL user. Substituting \eqref{transmitsignal_x} into \eqref{receivedsignal_yl}, we rewrite $y_{k}$ as
\begin{align}\label{receivedsignal_yl_new}
     y_k&={\h_{a, k}^\dagger \v_{k} s_{k}}+\sum\limits_{\ell=1}^L h_{\ell, k} z_\ell+ \sum\limits_{\substack{k^{\prime}=1\\
    k^{\prime} \neq k}}^K \h_{a, k}^\dagger \v_{k^{\prime}} s_{k^{\prime}}\nonumber\\
    &+\h_{a, k}^\dagger \s_r+ n_{k},~\forall k\in\K
\end{align}
where the first term represents the desired signal, the second term reflects the UL-to-DL interference, the third term is inter-user interference, and the fourth is the sensing interference. 

\subsubsection{Received signal model at the R-BS}

Following self-interference suppression, the signal received at the R-BS can be expressed as
\begin{align}\label{yA}
    \y_a&= \sum\limits^L_{\ell=1}\h_{\ell,a}z_\ell+ \y_r+\n_{SI} + \n_a,
\end{align}
where the first term represents the desired UL communications, the second term $\y_r$ is superimposed echo reflections from the M-UAV and clutter, the third term is residual self-interference (SI) modeled as $\n_{SI}\sim\CN(\0, \sigma^2_{SI} \I_{M_r})$, and the last term is AWGN at the R-BS with $\n_a \sim \CN(\0, \delta^2_a\I_{M_r})$. We assume that the Doppler effect arising from the movement of the M-UAV is constant during one timeslot and can be compensated for \cite{liu2018mimo}. Thus, $\y_r$ is expressed as
\begin{align}
  \y_r=  \zeta_0 \A_0\x + \mathbf{c},
\end{align}
where $\mathbf{c}=\sum^{N_c}\limits_{i=1}\zeta_i\A_i \x \sim \CN(\0, \mathbf{R}_c)$ indicates the clutter, whose covariance matrix $\mathbf{R}_c$ is assumed to be fixed and known, based on the classical iterative method \cite{chen2022generalized}.

\subsubsection{Received signal model at the M-UAV}
The signal received at the M-UAV is given by
\begin{align}\label{y_e}
    y_e = \sum\limits^L_{\ell=1} h_{\ell,e} z_\ell +\h^\dagger_{a,e}  \x+ n_e,
\end{align}
where the first term indicates the superimposed UL signals, the second term is the signal transmitted by the R-BS, and the last term is AWGN at the M-UAV with variance $\sigma^2_e$.

\subsection{Communication Secrecy Rate} 
Since the M-UAV is exposed to not only UL communication signals but also DL signals from the R-BS,  we define the worst-case communication secrecy rates (CSRs) of the system for DL and UL communications as
\begin{align}
    \dcsr &= \underset{k\in\K}{\min}\left[\log_2 (1+\gamma^\dl_k) - \log_2(1 + \gamma^\dl_{e,a})\right]^+ \\[5pt]
    \ucsr &= \underset{\ell\in\L}{\min}\left[\log_2 (1+\gamma^\ul_\ell) - \log_2(1 + \gamma^\ul_{e,\ell})\right]^+,
\end{align}
respectively, which depend on the SINRs defined next.
First, $\gamma^{\dl}_{k}$ indicates the SINR of DL user $k$, which can be obtained using \eqref{receivedsignal_yl_new} as
\begin{align}
\gamma^{\dl}_{k}
=\frac{\h_{a, k}^\dagger \V_{k} \h_{a, k}}{
    \hspace{-5mm}\sum\limits_{\ell} p_{\ell}|h_{\ell,k}|^2+\sum\limits_{\substack{k^{\prime}=1\\k^{\prime} \neq k}}^K \h_{a, k}^\dagger \V_{k^{\prime}} \h_{a, k}+\h_{a, k}^\dagger \W \h_{a, k}+\sigma_{k}^2},
\end{align}
where $\V_k =\v_k\v^\dagger_k$. Here, we apply receive beamforming vectors $\u_\ell,~\forall\ell\in \L$ at the R-BS to recover
UL communication symbols from \eqref{yA} through $\hat{s}_\ell = \u^\dagger_\ell{\y}_a$. Accordingly, the SINR for the \order{\ell} UL user  at the R-BS, indicated by $\gamma^\ul_\ell$, is calculated as
\begin{align}\label{gamma_ul_l}
    \gamma^\ul_\ell= \frac{p_\ell\u^\dagger_\ell \h_{\ell,a} \h^\dagger_{\ell,a}\u_\ell}
    {\u^\dagger_\ell\left( \sum\limits_{\ell'\neq \ell}p_{\ell'}\h_{\ell',a}\h^\dagger_{\ell',a} +|\zeta_0|^2\A_0\S\A^\dagger_0 + \boldsymbol{\Sigma}_n\right)\u_\ell},
\end{align}
where $\boldsymbol{\Sigma}_n = \mathbf{R}_c+(\sigma^2_{SI}+\sigma^2_a)\I_{M_r}$ is the aggregated noise covariance and  $\S = \E(\x\x^\dagger) = \V\V^\dagger + \W$ indicates the covariance of the transmit signal $\x$. The SINR of the \order{\ell} UL user at the M-UAV, treating the DL signals from the R-BS as interference, is $\gamma^\ul_{e, \ell}$, which can be obtained using \eqref{y_e} as
\begin{align}
    \gamma^\ul_{e, \ell} = \frac{p_\ell |h_{\ell, e}|^2}{\sum\limits^L_{\ell'=1, \ell'\neq \ell} p_{\ell'}|h_{\ell',e}|^2 +\h^\dagger_{a, e} \S \h_{a, e} +\sigma^2_e}.
\end{align}
Likewise, $\gamma^\dl_{e, a}$ denotes the SINR of the R-BS signal at the M-UAV treating the UL signals as interference:
\begin{align}
    \gamma^\dl_{e, a} = \frac{\h^\dagger_{a,e}\V\V^\dagger\h_{a,e}}{\sum\limits^L_{\substack{\ell=1}} p_{\ell}|h_{\ell,e}|^2 +\h^\dagger_{a,e}\W\h_{a,e}+\sigma^2_e}.
\end{align}

\subsection{Sensing Estimation Performance}
After successfully decoding the UL users' messages, \eqref{yA} can be expressed as
\begin{align}\label{yA_tilde}
    \y_a(t)= \zeta_0\A_0 \x(t-\tau) + \n(t),
\end{align}
where $\tau$ is the round-trip delay and $\n(t)$ is the aggregated noise, following a Gaussian distribution $\n(t)\sim \CN(\0, \boldsymbol{\Sigma}_n)$. The target position can be obtained by estimating the range $R$, azimuth $\theta$, and elevation $\phi$. In this work, we focus on range estimation assuming that the other parameters are known. 
The distance from the radar transmitter to the aerial target can be obtained as $R={C\tau}/{2}$, and thus is determined by an estimate of the delay. A typical delay estimator is given by \cite{kay1993fundamentals}
\begin{align}
    \tilde{\tau} = \underset{\tau}{\mathrm{argmax}}{\int  \y_a(t) \x(t-\tau)\mathrm{dt}.}
\end{align}

\begin{theorem}\label{theorem1}
The Cramér-Rao Lower Bound (CRLB) for estimation of the delay $\tau$ in \eqref{yA_tilde} is approximately
\begin{align}
    \crlb({\tau}) \approx \frac{1}{8\pi^2 B^2_{rms} |\zeta_0|^2 \tr( \A_0 \S \A^\dagger_0 \boldsymbol{\Sigma}^{-1}_n)},
\end{align}
where $B_{rms}$ is the root-mean-squared (RMS) bandwidth of the transmit signal such that $2\pi B_{rms}=\gamma B$, where $B$ is the full bandwidth and $\gamma$ is a scaling based on the radar waveform power spectrum, e.g., $\gamma^2={\pi^2}/{3}$ for a flat spectral shape. 
\end{theorem}
\begin{proof}
    See Appendix \ref{Appendix_A}.
\end{proof}

Using Theorem \ref{theorem1}, the CRLB for the target range estimate $R$ is approximately given by 
\begin{align}\label{crlb_range}
    \crlb(R) \approx \frac{C^2}{8\gamma^2 B^2 |\zeta_0|^2 \tr( \A_0 \S \A^\dagger_0 \boldsymbol{\Sigma}^{-1}_n)},
\end{align}
which should be kept below a specific threshold for adequate localization. Eq. \eqref{crlb_range} shows that the CRLB for range estimation is inversely proportional to the bandwidth of the transmit signal, the strength of the sensing channel matrix $\A_0$, and the aggregated noise covariance \( \boldsymbol{\Sigma}_n \). A stronger transmit signal (i.e., higher energy in $\dot{\x}$) and a more favorable environment (i.e., less interference, clutter, and noise power) reduce the CRLB, leading to more accurate range estimation. 

\begin{remark}
    In general, the R-BS may not have perfect knowledge of the M-UAV's DoA. Estimation errors in azimuth \(\theta\) and elevation \(\phi\) can lead to angular mismatch in beamforming, reducing the echo signal power and degrading localization accuracy, i.e., increasing the CRLB. These errors can lead to misalignment of the generated AN/jamming beams, which leads to reduced achievable secrecy performance for DL transmission. To address angular estimation errors, one can incorporate robust beamforming methods directly into the resource‐allocation framework by, for example, formulating a worst‐case design over an ellipsoidal or sector‐based angle uncertainty set. Alternatively, each transmission interval can be subdivided into a brief beam‐training subframe, using coarse scanning or grid‐refinement techniques to update the DoA estimate, followed by the secure ISAC transmission phase based on the refined angular information. A rigorous treatment of robustness enhancement is reserved for future investigation.
\end{remark}

\section{Game Theoretic Optimization Framework}\label{sec2_problem}

To capture the effect of the sequential decision-making in the dynamic and hierarchical ISAC system, we use game theory to model the strategic interactions between the legitimate network (Player 1) and the adversary (Player 2). Player 1 corresponds to the R-BS, which acts as the only decision-maker in the legitimate network, determining the necessary beamformers and dictating UL transmissions. This centralized approach is justified by the structure of the legitimate network, where the R-BS leverages system-wide knowledge and computational resources to ensure efficient coordination. Accordingly,  Player 1 seeks to optimize a given utility function by jointly designing the UL transmission power set $\{\p[n],~\forall n\}$, where $\p[n]\triangleq[p_\ell[n],~\forall \ell]$ is a row vector containing the transmit power for each UL user,  transmit beamformers $\{\V[n], \forall n\}$,  radar sensing/AN covariance matrices $\{\W[n], \forall n\}$, and receive beamformers $\{\U[n], \forall n\}$, where $\u_\ell[n]$ constitutes the \order{\ell} column of $\U[n]\in\C^{M_r\times L}$. Conversely, Player 2, represented by the M-UAV, designs the waypoints $\{\q_e[n],~\forall n\}$ of its trajectory to proactively improve its own utility function and degrade the performance of the legitimate network.

We model the interaction between the players as an NSG, where we assume that the M-UAV operates as the leader, planning its trajectory over the mission time to maximize its utility, balancing power conservation with mission objectives such as minimizing detection risk and reducing the achievable secrecy rates of the network. As the leader, the movement of the M-UAV (Player 2) directly influences the actions of the R-BS (Player 1), which acts as the follower. In response to Player 2’s behavior, Player 1 adjusts its resource allocation, minimizing transmission power consumption while meeting certain security and sensing requirements. To solve this hierarchical NSG, we employ a DRL approach, where Player 2 learns a trajectory that maximizes its utility over the mission duration, sequentially taking into account the responses from Player 1 in terms of the network’s resource optimization. Our simulation results suggest that this iterative process converges to a Stackelberg equilibrium, ensuring a stable outcome where neither player has an incentive to deviate unilaterally. This equilibrium reflects the leader-follower dynamics and fulfills the system’s power and security requirements.

\subsection{Resource Optimization for Player 1}
As mentioned above, we consider power efficiency as the utility/payoff for the legitimate network, where we aim to minimize the total power invested in UL transmissions, DL beamforming, and AN.  To improve the power efficiency, we seek to find the best power allocation for the legitimate network that satisfies the achievable secrecy rates of both the UL and DL users, while ensuring the radar sensing performance.  Thus, we formulate Player 1's utility at timeslot $n$ as
\begin{align}
   U_1[n] =\tr(\V[n]\V^\dagger[n]) + \tr(\W[n]) + \mathbf{1}^\top_L \p[n],
\end{align}
where the first term represents the power invested in DL beamforming, the second term is the aggregate transmit power allocated to the radar sensing/AN, and the last term is the total transmit power of the UL users. 
Player 1 aims to minimize total network power consumption  w.r.t. the decision variable set $\Xs_1=\{\X[n]\triangleq(\V[n], \W[n], \U[n], \p[n]),~\forall n\}$, given the strategy of Player 2, represented by $\Xs_2=\{\q_e[n],~\forall n\}$. Accordingly, the optimization problem for Player 1 can be formulated as
\begin{subequations}\label{cst_player1}
\begin{align}
\Prob{1}:&~ \stackrel{}{\underset{\Xs_1}{\mathrm{min}}} \sum\limits^N_{n=1} U_1[n]\nonumber\\
&~~\text{s.t.}~~ \ucsr[n] \geq \rho^\ul,~\forall n \label{ul_cst}\\
&~~~\quad~~ \dcsr[n] \geq \rho^\dl,~\forall n \label{dl_cst} \\
&~~~\quad~~ \crlb[n] \leq \rho_{est},~\forall n \label{er_cst} \\
&~~~\quad~~  \p[n]\succeq 0,~\W[n] \succeq 0,~\forall n 
\end{align}
\end{subequations}
where constraints \eqref{ul_cst} and \eqref{dl_cst} ensure the minimum required secrecy rates for the UL and DL communications, denoted by $\rho^\ul$ and $\rho^\dl$, respectively, and \eqref{er_cst} guarantees that the target range estimation variance does not exceed the threshold $\rho_{est}$.

\subsection{Trajectory Optimization for Player 2}

The M-UAV seeks to minimize its total power consumption to extend its lifetime while adversely impacting communication and sensing performance. As such, Player 2's utility at timeslot $n$ is expressed as a function of the optimal strategy of Player 1 (i.e., $\mathbf{X}^\star[n],~\forall n$), and is given by
\begin{align}
   U_2(\mathbf{X}^\star[n]) &=  \lambda U_1(\mathbf{X}^\star[n]) - (1-\lambda) P_f[n],
\end{align}
where $0 \leq \lambda \leq 1$ is a scaling factor to balance performance and power consumption. Thus, Player 2's optimization problem can be formulated as
\begin{subequations}
\begin{align}
\Prob{2}:&~ \stackrel{}{\underset{\Xs_2}{\mathrm{max}}}  \sum^N_{n=1} U_2(\mathbf{X}^\star[n])
\nonumber\\
&~~\text{s.t.}~~ \mathrm{C1}-\mathrm{C7}.
\end{align}
\end{subequations}

\subsection{Rationale for Leader-Follower Roles}
\label{subsec:rationale}

The assignment of the M-UAV as the Stackelberg leader and the R-BS as the follower is not arbitrary but hinges on three pillars:

\begin{itemize}
    \item \textbf{Proactive Nature of Adversarial Threat:} The M-UAV, as a mobile eavesdropper, engages proactively by optimizing its trajectory to exploit favorable AG channels while conserving power and evading detection. This mirrors real-world adversarial engagements in battlefield and surveillance ISAC scenarios, where threats dictate dynamics and compel reactive defense by legitimate networks.
    
    \item \textbf{Computational Tractability:} Letting the R-BS be the follower transforms its resource allocation problem $\Prob{1}$ into a Mathematical Program with Equilibrium Constraints (MPEC). Specifically, the convex problem $\Prob{1.4}$, detailed in Section~\ref{sec3_solution}, defines a strongly monotone variational inequality (VI) mapping \cite{Han2019}, ensuring a unique and stable solution for each fixed M-UAV trajectory. In contrast, modeling the R-BS as the leader would require solving an Equilibrium Program with Equilibrium Constraints (EPEC), which is known to exhibit non-unique equilibria, lacks convergence guarantees, and entails excessive computational overhead.
    
    \item \textbf{Existence and Uniqueness of Equilibria:} Under the current formulation, the leader's discrete trajectory space is compact, and the follower's best response satisfies the existence and uniqueness conditions of a Stackelberg equilibrium \cite{Han2019}. These properties are crucial for robust decision-making in dynamic ISAC environments.
\end{itemize}
Therefore, the selection of the M-UAV as the Stackelberg leader is both practically and theoretically well-founded, enabling a solvable, reliable, and operationally meaningful game-theoretic formulation.

\section{Non-cooperative Stackelberg Game Solution}\label{sec3_solution}

\subsection{Solution to \texorpdfstring{$\mathcal{P}1$}~ via Convex Optimization}
To tackle Player 1's optimization problem, we divide $\Prob{1}$ into $N$ subproblems and optimize each of them separately at each timeslot $n$. As such, we have the following reformulated subproblem:
\begin{subequations}\label{prob1_recast}
\begin{align}
\Prob{1.1}:&~ \stackrel{}{\underset{\Xs_1}{\mathrm{min}}} \tr(\V\V^\dagger) + \tr(\W) + \mathbf{1}^\top_L \p \nonumber\\
&~~\text{s.t.}~~ \ucsr \geq \rho^\ul, \label{cst2}\\
&~~~\quad~~ \dcsr \geq \rho^\dl, \label{cst1}\\
&~~~\quad~~ \crlb \leq \rho_{est}, \label{cst3}\\
&~~~\quad~~  \W \succeq \mathbf{0},~\p \succeq \mathbf{0}. \label{cst_psd}
\end{align}
\end{subequations}

To introduce an iterative algorithm to solve $\Prob{1.1}$ in a more tractable way, we introduce the slack variable vector $\pmb{\mu} = \{\mu_i, i=1,..K+1\}$ and rewrite constraint \eqref{cst1} as
\begin{subequations}\label{cst_dl}
\begin{align}
    &\gamma^\dl_k \geq \mu_k,~\forall k\in\K\\
    &\gamma^\dl_{e,a} \leq \mu_{K+1},\\
    &\log_2(1+\mu_k) - \left[\log_2(1+\mu^{lo}_{K+1}) + \frac{\mu_{K+1}-\mu^{lo}_{K+1}}{\ln 2 (1+\mu^{lo}_{K+1})}\right]\nonumber\\
    &\geq \rho^{\dl},~\forall k\in\K \label{cst36c_cvx}
\end{align}\end{subequations}
where $\mu^{lo}_{K+1}$ indicates the local value of $\mu_{K+1}$ at each iteration. Note that the last inequality in \eqref{cst_dl} is convex and obtained via a first-order Taylor approximation. Similarly, constraint \eqref{cst2} is reformulated by introducing the slack variable vector $\pmb{\omega} =\{\omega_i,~ i=1\cdots, 2L\}$  as
\begin{subequations}\label{cst_ul}\begin{align}
    &\gamma^\ul_\ell \geq \omega_\ell,~\forall \ell\in\L\\
    &\gamma^\ul_{e,\ell} \leq \omega_{L+\ell}, ~\ell\in\L\label{cst_ul2}\\
    &\log_2(1+\omega_\ell) - \left[\log_2(1+\omega^{lo}_{L+\ell}) + \frac{\omega_{L+\ell}-\omega^{lo}_{L+\ell}}{\ln 2 (1+\omega^{lo}_{L+\ell})}\right]\nonumber\\
    &\geq \rho^{\ul},~\forall\ell\in\L\label{cst37c_cvx}
\end{align}\end{subequations}
where $\omega^{lo}_{L+\ell}$ indicates the local value of $\omega_{L+\ell}$ at each iteration. Thus, $\Prob{1.1}$ can be rewritten as

\begin{subequations}\label{prob1_recast2}
\begin{align}
\Prob{1.2}:&~ \stackrel{}{\underset{\Xs_1, \pmb{\mu}, \pmb{\omega}}{\mathrm{min}}} \tr(\V\V^\dagger) + \tr(\W) + \mathbf{1}^\top_L \p \\
&~~\text{s.t.}~~ \eqref{cst_dl}, \eqref{cst_ul}, \eqref{cst3}, \eqref{cst_psd}.
\end{align}
\end{subequations}

Since $\{\U[n],~\forall n\}$ only appears in the expression for $\gamma^\ul_\ell$ in \eqref{gamma_ul_l}, we first maximize each of these quantities. Note that this also indirectly reduces the transmit power required for satisfying the UL constraint \eqref{cst_ul}. In light of this, we need to solve the following optimization w.r.t. $\u_\ell,~\forall \ell$:
\begin{align}\label{GRayleighProb}
    \underset{\u_\ell}{\arg\max} \frac{p_\ell\u^\dagger_\ell \h_{\ell,a} \h^\dagger_{\ell,a}\u_\ell}
    {\u^\dagger_\ell\left( \sum\limits_{\ell'\neq \ell}p_{\ell'}\h_{\ell',a}\h^\dagger_{\ell',a} +|\zeta_0|^2\A_0\S\A^\dagger_0 + \boldsymbol{\Sigma}_n\right)\u_\ell}.
\end{align}
We note that \eqref{GRayleighProb} is in the form of a generalized Rayleigh quotient $\Rq(\B,\D;\x)\triangleq\frac{\x^\dagger \B \x}{\x^\dagger \D \x}$, where $\x$ is a non-zero vector, $\B$ is a Hermitian matrix, and $\D$ is a PSD matrix. Thus, the maximum value of $\Rq(\B, \D;\x)$  equals the largest generalized eigenvalue of $(\B, \D)$, achieved by the corresponding generalized eigenvector. However, since the numerator in \eqref{GRayleighProb}, i.e., $\B = \h_{\ell,a} \h^\dagger_{\ell,a}$, is a rank-one matrix, the solution to \eqref{GRayleighProb} can be obtained  in closed-form as
\begin{align}
    \u^*_\ell = \left( \sum\limits_{\ell'\neq \ell}p_{\ell'}\h_{\ell',a}\h^\dagger_{\ell',a} +|\zeta_0|^2\A_0\S\A^\dagger_0 + \boldsymbol{\Sigma}_n\right)^{-1} \h_{\ell,a},~\forall \ell.
\end{align}
Substituting $\u^*_\ell,~\forall \ell$ into the respective SINR expression of the \order{\ell} UL user, \eqref{gamma_ul_l} becomes
\begin{align}
    \widetilde{\gamma^\ul_\ell} &=\nonumber\\
    &\hspace{-5mm}p_\ell \h^\dagger_{\ell,a}\left( \sum\limits_{\ell'\neq \ell}p_{\ell'}\h_{\ell',a}\h^\dagger_{\ell',a} \hspace{-0.5mm}+\hspace{-0.5mm}|\zeta_0|^2\A_0\S\A^\dagger_0 + \boldsymbol{\Sigma}_n\hspace{-0.5mm}\right)^{-1}\hspace{-2mm}\h_{\ell,a}.
\end{align}
Thus, we can reformulate $\Prob{1.2}$ irrespective of $\U$ as
\begin{subequations}\label{prob1_recast3}
\begin{align}
\Prob{1.3}:&~ \stackrel{}{\underset{\{\V_k\}, \W, \p, \pmb{\mu}, \pmb{\omega}}{\mathrm{min}}} \sum^K_{k=1}\tr(\V_k) + \tr(\W) + \mathbf{1}^\top \p_\ell \nonumber\\
&~~\text{s.t.}~~ \eqref{cst_dl}, \eqref{cst3}, \eqref{cst_psd}, \eqref{cst_ul2}, \eqref{cst37c_cvx}\\
&~~~\quad~~ \widetilde{\gamma^\ul_\ell} \geq \omega_\ell,~\forall \ell \label{cst42b}\\
&~~~\quad~~ \V_k \succeq \mathbf{0},~\forall k\label{cst42c}\\
&~~~\quad~~ \rank(\V_k)=1,~\forall k\label{cst12_rnk1} .
\end{align}
\end{subequations}
The first term of $\Prob{1.2}$ has been replaced with the equivalent sum $\sum^K_{k=1}\tr(\V_k)$, where $\V_k\triangleq\v_k\v^\dagger_k$, to make it convex. Thus, constraint \eqref{cst12_rnk1} is imposed since the $\V_k$ are all rank-one matrices. This constraint will be dealt with below via semidefinite relaxation (SDR) \cite{Luo2010}.

Since the objective function of $\Prob{1.3}$ is now convex, we focus on the non-convex constraints. The first constraint in \eqref{cst_dl} can be rewritten by introducing a slack variable vector $\mathbf{t}=\{t_1, \cdots, t_{K+1}\}$, as
\begin{subequations}\label{cst43_cvx}
\begin{align}
    &\sum\limits_{\ell} p_{\ell}|h_{\ell,k}|^2 + \tr\left(\left[\sum\limits_{\substack{k^{\prime}=1\\k^{\prime} \neq k}}^K\V_{k^{\prime}} + \W\right]\H_{a,k}\right) + \sigma^2_k \nonumber\\
    &\leq  f_1(t_k,\mu_k; t^{lo}_k,\mu^{lo}_k),~\forall k\in\K\\
&t^2_k \leq\tr(\V_k \H_{a,k}),~\forall k\in\K
\end{align}
\end{subequations}
where   $\H_{a,k} \triangleq \h_{a,k} \h^\dagger_{a,k},~\forall k$, and the function $f_1$ is defined at the local point ($t^{lo}_k, \mu^{lo}_k$) as
\begin{align}
\hspace{-3mm} f_1(t_k,\mu_k; t^{lo}_k,\mu^{lo}_k) \triangleq    \frac{t^{lo}_k}{(\mu^{lo}_k)^2} \left(2\mu^{lo}_k t_k - t^{lo}_k \mu_k \right),~\forall k\in\K.
\end{align}

We reformulate the second constraint of \eqref{cst_dl} as a convex approximation at a given local point $t^{lo}_{K+1}$, leading to
\begin{subequations}\label{cst45_cvx}
    \begin{align}
    &\frac{t^2_{K+1}}{\mu_{K+1}}\leq {\sum\limits^L_{\substack{\ell=1}} p_{\ell}|h_{\ell,e}|^2 +\tr(\W\H_{a,e}) +\sigma^2_e},\\
    &\sum^K_{k=1}\tr(\V_k \H_{a,e}) \leq -(t^{lo}_{K+1})^2 + 2t_{K+1}(t^{lo}_{K+1}),
\end{align}
\end{subequations}
where $\H_{a,e} \triangleq \h_{a,e}\h^\dagger_{a,e}$. We note that the first inequality in \eqref{cst45_cvx} is a convex constraint since the left-hand-side (LHS) is a quadratic-over-linear term and the right-hand-side (RHS) is an affine expression w.r.t. the optimization variables. The second inequality is also convex due to the first-order relaxation. 

The non-convex constraint \eqref{cst_ul2} is reformulated to be convex by introducing a slack variable vector $\mathbf{s}=\{s_\ell, \forall \ell \in \L\}$:
\begin{subequations}\label{cst46_cvx}
\begin{align}
    &s^2_\ell \leq {\sum\limits^L_{\substack{\ell'=1\\ \ell'\neq \ell}} p_{\ell'}|h_{\ell',e}|^2 + \tr(\S\H_{a,e}) +\sigma^2_e},~\forall \ell \in \L\\
    &\frac{1}{\omega_{L+\ell}} \leq f_1 (s_\ell, \tilde{p}_\ell; s^{lo}_\ell, \tilde{p}^{lo}_\ell),~\forall \ell \in \L
\end{align}
\end{subequations}
where $\tilde{p}_\ell = p_\ell|h_{\ell, e}|^2$. 

The constraint \eqref{cst3} is rewritten as 
\begin{align}\label{cst47_cvx}
    \frac{C^2}{8\gamma^2 B^2 \rho_{est} }\leq \tr( \boldsymbol{Y} \S \boldsymbol{Y}^\dagger),
\end{align}
where $\boldsymbol{Y} \triangleq |\zeta_0|\boldsymbol{\Sigma}^{-\frac{1}{2}}_n \A_0$, and this constraint is convex since the trace function is affine in $\S$ and the sum of PSD matrices preserves convexity. 

We address the non-convexity of \eqref{cst42b} by defining
\[\Psi_\ell(\p,\S) \triangleq \sum\limits_{\ell'\neq \ell}p_{\ell'}\h_{\ell',a}\h^\dagger_{\ell',a} +|\zeta_0|^2\A_0\S\A^\dagger_0 + \boldsymbol{\Sigma}_n,~\forall \ell\in \L\]
and rewriting \eqref{cst42b} as
\begin{align}\label{cst48}
     \h^\dagger_{\ell,a} \Psi^{-1}_\ell(\p,\S)\h_{\ell,a} \geq \frac{\omega_\ell}{p_\ell},~\forall \ell\in \L.
\end{align}
Obviously, the LHS of \eqref{cst48} is convex w.r.t. $\Psi_\ell(\p,\S)$. Thus, following a first-order Taylor approximation, we obtain a global lower-bound at the local point $\Psi^{lo}_\ell(\p ^{lo}, \S^{lo})$ as
\begin{align}\label{cst49}
    \h^\dagger_{\ell,a} &\Psi^{-1}_\ell\h_{\ell,a} \geq 
    \h^\dagger_{\ell,a} [\Psi^{lo}_\ell]^{-1}\h_{\ell,a} \nonumber\\
    &- \h^\dagger_{\ell,a} [\Psi^{lo}_\ell]^{-1}(\Psi_\ell - \Psi^{lo}_\ell)[\Psi^{lo}_\ell]^{-1}\h_{\ell,a},~\forall \ell\in\L.
\end{align}
Using \eqref{cst49}, we reformulate \eqref{cst48} by introducing non-negative slack variables $\pmb{\iota}=\{\iota_l, \forall l\in \L\}$ as
\begin{subequations}\label{cst50_cvx}
\begin{align}
    &\Xi_\ell (\iota_\ell, \Psi_\ell)  \leq 0,~\forall \ell\in\L\\
    &\omega_\ell \leq -(\iota^{lo}_\ell)^2 +2\iota^{lo}_\ell\iota_\ell,~\forall \ell\in\L
\end{align}
\end{subequations}
where 
\begin{align}
    \Xi_\ell (\iota_\ell, \Psi_\ell) &\triangleq\frac{\iota^2_\ell}{p_\ell} + \h^\dagger_{\ell,a} [\Psi^{lo}_\ell]^{-1}\Psi_\ell [\Psi^{lo}_\ell]^{-1}\h_{\ell,a} \nonumber\\&
    - 2\h^\dagger_{\ell,a} [\Psi^{lo}_\ell]^{-1}\h_{\ell,a},~\forall \ell\in\L.
\end{align}
Notice that $\Xi_\ell$ is a convex function w.r.t. optimization variables $\pmb{\iota}, \p, \V, \W$. The first derivatives of $\Xi_\ell$ are as follows:
\begin{align}
    &\frac{\partial  \Xi_\ell }{\partial \iota_\ell} = \frac{2\iota_\ell}{p_\ell},~ \frac{\partial  \Xi_\ell }{\partial p_\ell} = \frac{-\iota^2_\ell}{p_\ell},~\forall \ell\in\L\nonumber\\
    &\frac{\partial  \Xi_\ell }{\partial \p_{\ell'}} = \h^\dagger_{\ell,a} [\Psi^{lo}_\ell]^{-1}\h_{\ell', a}\h^\dagger_{\ell', a}[\Psi^{lo}_\ell]^{-1}\h_{\ell,a},~\forall \ell' \in \L\setminus \{\ell\}  \nonumber\\
    &\frac{\partial  \Xi_\ell }{\partial \V_k} = \frac{\partial  \Xi_\ell}{\partial \W} = |\zeta_0|^2\h^\dagger_{\ell,a} [\Psi^{lo}_\ell]^{-1}\A_0\A^\dagger_0[\Psi^{lo}_\ell]^{-1}\h_{\ell,a},~\forall \ell\in\L. \nonumber
\end{align}
The second derivatives of the twice-differentiable function $\Xi_\ell$ are zero, except $\frac{\partial^2  \Xi_\ell }{\partial \iota^2_\ell} = \frac{2}{p_\ell}$ and $\frac{\partial^2  \Xi_\ell }{\partial p_\ell} = \frac{\iota^2_\ell}{p^2_\ell}$. This implies that the Hessian $\nabla^2  \Xi_\ell$ is an all-zero matrix except for the \order{\ell} block-diagonal entries which are non-negative. Thus, $\nabla^2  \Xi_\ell$ is PSD and $\Xi_\ell$ is convex.  

Finally, following the SDR approach, we drop the rank-one constraint in $\Prob{1.3}$, leading to the convex optimization problem
\begin{subequations}\label{prob1_cvx}
\begin{align}
\Prob{1.4}:&~ \stackrel{}{\underset{\V, \W, \p, \pmb{\mu}, \pmb{\omega}, \mathbf{t}, \mathbf{s}, \pmb{\iota}}{\mathrm{min}}} \sum^K_{k=1}\tr(\V_k) + \tr(\W) + \mathbf{1}^\top_L \p \nonumber\\
&~~\text{s.t.}~~ \eqref{cst43_cvx}, \eqref{cst45_cvx}, \eqref{cst36c_cvx}, \eqref{cst47_cvx}, \eqref{cst_psd},\nonumber\\
&~~~\quad~~ \eqref{cst46_cvx}, \eqref{cst37c_cvx}, \eqref{cst50_cvx}, \eqref{cst42c}.\nonumber
\end{align}
\end{subequations}
Problem $\Prob{1.4}$ can be efficiently solved by standard convex optimization tools such as CVX \cite{CVXResearch2012}.
\textcolor{black}{In general, post-processing techniques such as Gaussian randomization\,\cite{Luo2010} or rank-minimization methods\,\cite{TatarMamaghani2022} are needed to recover the rank-one solution from the SDR. Nevertheless, for Problem $\Prob{1.4}$, the relaxation is \textit{provably tight} (see Appendix~\ref{app:rank1} for detailed proof); consequently, the relaxed solution is indeed optimal.} We now summarize the full solution to Player 1's resource optimization problem in Algorithm \ref{algo1}.

\begin{algorithm}
\caption{\label{algo1} Convex Solution to Resource Optimization}
\begin{algorithmic}[1]
    \State \textbf{Input:} Convergence threshold \( \varepsilon \)
    \State \textbf{Initialize:} 
    \Statex \quad Set \( \V^{lo} \) and \( \W^{lo} \) as normalized rank-one Hermitian PSD matrices with some positive power vector \( \p^{lo} \).  
    \Statex \quad Initialize feasible slack variables \( \pmb{\mu}, \pmb{\omega}, \mathbf{t}, \mathbf{s}, \pmb{\iota} \).
    \Statex Solve the following feasibility problem using CVX
    \begin{align*}
       & \text{Find:} \quad  (\V,  \W, \p) \\
        &\text{subject to: } \quad \eqref{cst43_cvx}, \eqref{cst45_cvx}, \eqref{cst36c_cvx}, \eqref{cst47_cvx}, \eqref{cst_psd}, \\
        & \quad\quad\quad \eqref{cst46_cvx}, \eqref{cst37c_cvx}, \eqref{cst50_cvx}, \eqref{cst42c}.
    \end{align*}
    \Statex \quad Extract the solution as initialization for \(\Prob{1.4}\)   
    \Statex \quad Calculate \(U^{(0)}_1\) and set iteration index \( i \gets 1 \)  
    \Repeat
        \State \quad Solve the convex problem \( \Prob{1.4} \).

        \State \quad Extract solution \( \V^{(i)}, \W^{(i)}, \p^{(i)}, \pmb{\mu}^{(i)}, \pmb{\omega}^{(i)}, \mathbf{t}^{(i)}\), \(\mathbf{s}^{(i)}, \pmb{\iota}^{(i)} \). Then  calculate \( U^{(i)}_1 \) and the fractional error: 
        \[f_{err} = \frac{|U^{(i-1)}_1 - U^{(i)}_1|}{U^{(i-1)}_1},\]
        
        \State \quad \(i \gets i + 1\).
        \Until{Convergence criterion is met: $f_{err} \leq \varepsilon$}
    
    \State \textbf{Output:} Optimal solutions \( \V^\star \), \( \W^\star \), \( \p^\star \).
\end{algorithmic}
\end{algorithm}

\subsection{Solution to \texorpdfstring{$\mathcal{P}2$}~ via DRL}
To address the optimization problem for Player 2, we employ an advanced reinforcement learning algorithm, where Player 2 learns to adaptively optimize its trajectory over the mission duration in response to the near-optimal strategies of Player 1. Given the inherent complexity and dynamics of the problem, DRL is a suitable choice for the following reasons.

\begin{itemize}
    \item \textbf{Dynamic and sequential interactions:} The interaction between Player 1 and Player 2 follows a sequential decision-making process, typical of a Stackelberg game. DRL is well suited to model such sequential decision-making processes by iteratively learning optimal policies based on environmental feedback.

    \item \textbf{High-dimensional and non-convex strategy space:} While Player 1’s resource allocation is solved using a low-complexity SCA approach, optimizing the M-UAV’s trajectory involves navigating a high-dimensional, non-convex space. DRL effectively handles such complexity without requiring convexity assumptions.  

    \item \textbf{Adaptability to dynamic environments:} DRL operates in a model-free manner, making it robust against changing channel conditions, adversary behavior, and other dynamic factors in the ISAC system.  

    \item \textbf{Scalability and real-time efficiency:} Unlike traditional optimization methods that can be computationally prohibitive, DRL leverages neural networks to approximate policies, enabling real-time trajectory planning and scalability to large-dimensional systems.  
\end{itemize}
These advantages make DRL a compelling choice for solving problem $\Prob{2}$ by enabling the M-UAV to iteratively optimize its trajectory while accounting for Player 1’s strategies. Our results indicate that the solution converges, ensuring a stable and efficient outcome.  The following sections detail the components necessary for implementing the DRL solution.

\subsection{MDP Framework}
We model the M-UAV’s trajectory optimization as a Markov Decision Process (MDP). The MDP serves as the foundation for training the DRL agent, establishing a relationship
between interaction-based learning and goal achievement \cite{Sutton2018}, and is typically represented by the tuple 
$(\mathcal{S}, \mathcal{A}, \mathtt{r}, \gamma)$, whose components are defined below.
\subsubsection{{State space} \texorpdfstring{\( \mathcal{S} \)}{State Space}}
 Each state \( \mathtt{s}_n \in \mathcal{S} \) encodes the M-UAV’s 3D position \(\q_e[n]\) at time step \( n \), which satisfies constraints $\mathrm{C2}-\mathrm{C3}$, remaining mission time  \(t_{\mathrm{rem}} = T - n\delta_t\), and the last observed resource allocation from the legitimate network \(\mathbf{X}^\star_1[n]\). The state space is defined as  $\mathcal{S}=\{\q_e[n], t_{\mathrm{rem}}[n], \mathbf{X}^\star_1[n],~\forall n\}$, which is a finite set since the M-UAV locations are defined on a grid of size $N_{grid}\times N$, which depends on the locations and remaining time steps, i.e., \( \mathtt{s}_n = \{\q_e[n], t_{\mathrm{rem}}[n]\} \), since $\mathbf{X}^\star_1[n],~\forall n$ is a deterministic function of $\q_e[n]$ for the given channel realization.

\subsubsection{{Action space} \texorpdfstring{\( \mathcal{A} \)}{A}} The M-UAV’s action \( \mathtt{a}_n \in \mathcal{A} \) at each time step consists of its motion in six possible directions—up ($+z$), down ($-z$), forward ($+x$), backward ($-x$), left ($-y$), right ($+y$)—with three discrete speed levels for each: Low (\(\mathtt{v_{low}}\)), medium (\(\mathtt{v_{med}}\)), and high (\(\mathtt{v_{high}}\)). Thus, \( \mathcal{A} = \mathtt{\{dir_i, vel_j,~\forall i,j}\} \), where \( \mathtt{dir=\{U, D, F, B, L, R\}} \) represents the direction set and \( \mathtt{vel=\{v_{low}, v_{med}, v_{high}\} }\) is the speed set. The speed levels are chosen such that the action set $\mathcal{A}$ satisfies constraints $\mathrm{C4}-\mathrm{C7}$. 
In the deterministic environment considered, the state transition probabilities satisfy $\Pr(\mathtt{s}'|\mathtt{s}, \mathtt{a})=1$, meaning the next state is fully determined by the current state and action.

\subsubsection{{Reward function} \texorpdfstring{\( \mathtt{r}(\mathtt{s}, \mathtt{a}):\mathcal{S}\times\mathcal{A}\rightarrow \mathbb{R} \)}{R}}
To balance the main objective of maximizing \( U_2 \) with the secondary goals of encouraging progress towards the final spherical region location with radius $r_{f}$ centered at \( \q_f \), avoiding obstacles, and completing the mission within the allowed time \( T \), we design the reward function as described below.
\begin{enumerate}[label=(\roman*)]
    \item {Primary objective reward:} The core of the reward function is based on the objective function of $\Prob{2}$, given by
    \begin{align}
        \mathtt{r}_\mathrm{main} = \left(\frac{t_{\mathrm{rem}}}{T}\right) U_2(\mathbf{X}^\star(\q_{\mathtt{next}})).
    \end{align}
    \item {Distance-based penalty:} To encourage intermediate progress towards the goal, we define the following distance-based penalty term:
    \begin{align}
        \mathtt{p}_\mathrm{dist} = - \left(1 + \frac{T - t_{\mathrm{rem}}}{T}\right) \frac{\| \q_{\mathrm{next}} - \q_f \|}{\| \q_i - \q_f \|}.
    \end{align}
    \item {Boundary violation and obstacle collision penalty:} To ensure the M-UAV operates in the permitted region $\R_o$ and avoids collisions with environmental obstacles, a penalty $\mathtt{p}_\mathrm{fail}$ is applied if an obstacle is hit or the agent attempts to move outside the region. This penalty discourages the M-UAV from choosing risky trajectories that could lead to collisions, thus promoting safer navigation.
    \item {Goal completion reward:} If the M-UAV reaches the terminal region by the end of the mission time, i.e.,  $t_{rem}\geq0$, it receives a bonus  $\mathtt{r}_\mathrm{goal}$ to encourage satisfaction of constraint $\mathrm{C1}$.
\end{enumerate}
Thus, the instantaneous reward \( \mathtt{r} \) at time step $n$ is 
\begin{align}\label{reward}
\mathtt{r}_n &= \eta_1 \mathtt{r}_\mathrm{main}  + \eta_2 \mathtt{p}_\mathrm{dist} \nonumber\\
&+
\begin{cases} 
\mathtt{r}_\mathrm{goal}, & \text{if}~~\|\q_{\mathrm{next}}-\q_f\| \leq r_f~\text{and}~t_{\mathrm{rem}} \geq 0, \\ 
\mathtt{p}_\mathrm{fail}, & \text{if}~~\|\q_{\mathrm{next}}-\q_f\| > r_f~\text{and}~t_{\mathrm{rem}} = 0, \\ 
\mathtt{p}_\mathrm{fail}, & \text{if}~~\q_{\mathrm{next}}\notin \R_o, \\ 
\mathtt{p}_\mathrm{fail}, & \text{if}~~\|\q_{\mathrm{next}}-\{\q_o\}\| < d^{\mathrm{min}}, \\0, & \text{otherwise}
\end{cases}
\end{align}
where $\eta_1$ and $\eta_2$ are chosen scaling factors.

\subsubsection{{Discount factor} \texorpdfstring{$\gamma$}{gamma}} 
The discount factor \(\gamma\) controls the trade-off between immediate and long-term rewards. Higher values of \(\gamma\) place more emphasis on future rewards, encouraging long-term strategies, while lower values prioritize immediate rewards. The agent's objective is to maximize the discounted cumulative reward (return):
 \begin{align}G_n = \sum_{k=0}^\infty \gamma^k \mathtt{r}_{n+k}, \end{align}
where \( \gamma \in [0, 1) \) is the discount factor.

The MDP framework\footnote{The grid‐based, deterministic MDP studied here is intended as a proof‐of‐concept. Future work will consider extensions to continuous action‐spaces under stochastic and partially observable dynamics, incorporate actor–critic DRL methods, and employ generative learning for adaptive reward shaping to generalize across diverse secure ISAC scenarios.} characterizes the environment in which the DRL agent is trained, and the agent's goal is to learn the optimal policy that describes transitions between the current state and the next state based on the action set. The agent aims to maximize the expected return, which is measured by the discounted cumulative reward. 

\subsection{DRL Agent Framework}
In this work, we employ an advanced DRL architecture, known as the Double Deep Q-Network (DDQN) \cite{van2016deep}, to learn the transition policy, and thus optimize the M-UAV’s trajectory in the considered MDP framework.  The DDQN architecture combines neural networks for effective learning with experience replay and target networks for stable updates. This can enable the DRL agent to converge to a Stackelberg equilibrium, enabling a stable and efficient solution to the hierarchical game while meeting the system's power and security requirements. In what follows, we detail the adopted DDQN agent and the training process.

\subsubsection{Double DQN agent}
DDQN is a model-free off-policy reinforcement learning algorithm designed for environments with discrete action spaces \cite{van2016deep}. The DDQN agent trains a Q-value function that estimates the expected discounted cumulative long-term reward associated with following an optimal policy. A key innovation of DDQN is its ability to mitigate the overestimation bias commonly encountered in traditional Q-learning approaches \cite{mamaghani2020intelligent}. This is accomplished by decoupling action selection from action evaluation, leading to a more stable learning environment.

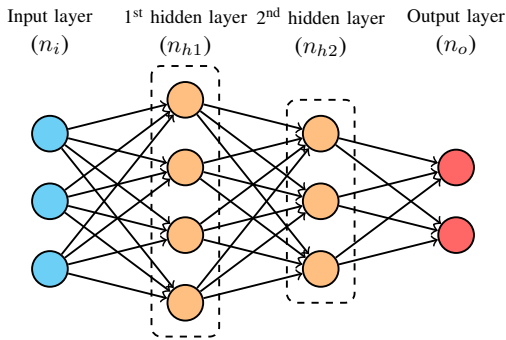
\begin{figure}[t]
    \centering
        \scalebox{0.9}{
    \begin{tikzpicture}[
        neuron/.style={circle, draw, minimum size=15pt, inner sep=0pt},
        input neuron/.style={neuron, fill=cyan!50},
        hidden neuron/.style={neuron, fill=orange!50},
        output neuron/.style={neuron, fill=red!60},
        layer/.style={minimum width=2.5cm},
        ->, thick
    ]
        \node[input neuron] (I-1) at (0, 1) {};
        \node[input neuron] (I-2) at (0, 0) {};
        \node[input neuron] (I-3) at (0, -1) {};
        \node[align=center] at (0, 2.5) {\footnotesize Input layer \\ ($n_i$)};

        \node[hidden neuron] (H1-1) at (2, 1.5) {};
        \node[hidden neuron] (H1-2) at (2, 0.5) {};
        \node[hidden neuron] (H1-3) at (2, -0.5) {};
        \node[hidden neuron] (H1-4) at (2, -1.5) {};
        \node[align=center] at (2, 2.5) {\footnotesize 1\textsuperscript{st} hidden layer \\ ($n_{h1}$)};

        \draw[dashed, rounded corners] (1.5, 2) rectangle (2.5, -2);

        \node[hidden neuron] (H2-1) at (4, 1) {};
        \node[hidden neuron] (H2-2) at (4, 0) {};
        \node[hidden neuron] (H2-3) at (4, -1) {};
        \node[align=center] at (4, 2.5) {\footnotesize 2\textsuperscript{nd} hidden layer \\ ($n_{h2}$)};
        \draw[dashed, rounded corners] (3.5, 1.5) rectangle (4.5, -1.5);

        \node[output neuron] (O-1) at (6, -0.5) {};
        \node[output neuron] (O-2) at (6, 0.5) {};
        \node[align=center] at (6, 2.5) {\footnotesize Output layer \\ ($n_o$)};

        \foreach \i in {1,...,3} {
            \foreach \j in {1,...,4} {
                \draw[->] (I-\i) -- (H1-\j);
            }
        }

        \foreach \i in {1,...,4} {
            \foreach \j in {1,...,3} {
                \draw[->] (H1-\i) -- (H2-\j);
            }
        }

        \foreach \i in {1,...,3} {
             \foreach \j in {1,2} {
                \draw[->] (H2-\i) -- (O-\j);
            }
        }

\end{tikzpicture}}
\caption{Illustration of the DNN architecture adopted in the DDQN framework.}
\label{fig:nn_structure}
\end{figure}

The architecture of the DDQN comprises two separate neural networks: the primary Q-network (online network) and the target network. The primary Q-network is responsible for selecting actions based on the current state, while the target network is employed to provide stable Q-value updates. Both networks are structured as deep neural networks (DNNs), designed to approximate the Q-value function, denoted as \( Q(\mathtt{s}, \mathtt{a}; \pmb{\vartheta}) \), where \(\mathtt{s} \) represents the M-UAV’s current state, \( \mathtt{a} \) is the selected action, and \( \pmb{\vartheta} \) encapsulates the network parameters. As illustrated in Fig. \ref{fig:nn_structure}, the fully-connected network consists of an input layer with \( n_i \) neurons representing the dimensions of the state/feature space, two hidden layers with \( n_{h1} \) and \( n_{h2} \) neurons, each employing ReLU activation functions to introduce non-linearity and enable the network to learn complex state-action mappings with minimal complexity, and an output layer with \( n_o \) neurons, corresponding to the number of possible actions. The target network shares the same architecture but has parameters \( \pmb{\vartheta}^- \) that are updated less frequently, stabilizing the training process by providing consistent Q-values throughout the iterations.

A key aspect of DDQN is its two-step action evaluation mechanism. In each training step, the online network selects the action \( \mathtt{a}' \) using the $\epsilon$-greedy approach to balance exploration and exploitation, defined as
   \begin{align}\label{epsilon-greedy}
\hspace{-3mm}\mathtt{a}' = 
    \begin{cases} 
    \arg\max\limits_{\mathtt{a} \in \mathcal{A}} Q(\mathtt{s}, \mathtt{a}; \pmb{\vartheta}), & \text{with probability } 1 - \epsilon, \\
    \text{a random action}\in\mathcal{A}, & \text{with probability } \epsilon.
    \end{cases}
    \end{align}
The target network then evaluates this action by computing the Q-value \( Q(\mathtt{s}', \mathtt{a}'; \pmb{\vartheta}^-) \). This helps reduce the risk of overestimating the Q-values and enhances the stability of the learning process.

The loss function for the DDQN is formulated as
\begin{align}\label{loss_function}
\hspace{-3mm}L(\pmb{\vartheta}) = \mathbb{E}_{(\mathtt{s}, \mathtt{a}, \mathtt{r}, \mathtt{s}')} \left( \left[ \mathtt{r} + \gamma Q(\mathtt{s}', \mathtt{a}'; \pmb{\vartheta}^-) - Q(\mathtt{s}, \mathtt{a}; \pmb{\vartheta}) \right]^2 \right).
\end{align}
Using \( \pmb{\vartheta}^- \) for the Q-value evaluation, the DDQN effectively minimizes the variance of the Q-values. 
To further enhance sample efficiency in the training data, the DDQN employs an experience replay buffer. This buffer stores past transitions \( (\mathtt{s}, \mathtt{a}, \mathtt{r}, \mathtt{s}') \), allowing the agent to randomly sample mini-batches during each training episode. This methodology enables the agent to learn from a diverse set of experiences, reducing the likelihood of overfitting to recent transitions. 
Moreover, the parameters of the target network \( \pmb{\vartheta}^- \) are periodically synchronized with those of the online network \( \pmb{\vartheta} \) at a predefined update frequency \( f^{\mathtt{update}} \). This strategy provides a slowly moving target for Q-value updates, which is crucial for stabilizing the learning process and further reducing the variance in the target Q-values. Together, these components form a robust framework for the DDQN agent, facilitating effective trajectory optimization in dynamic environments.

The training process of our proposed Stackelberg game problem is summarized in Algorithm \ref{algo2}.

\begin{algorithm}
\caption{\label{algo2} DDQN Solution for Trajectory Optimization}
\begin{algorithmic}[1]
    \State \textbf{Input:} Initial Q-network parameters \( \pmb{\vartheta} \), experience replay buffer \( \mathcal{D} \), exploration rate \( \epsilon \), discount factor \( \gamma \), and maximum number of  training episodes $\mathrm{epMax}$
    
    \State \textbf{Initialize:} Q-network \( Q(\mathtt{s}, \mathtt{a}; \pmb{\vartheta}) \) with random weights, target network \( Q(\mathtt{s}, \mathtt{a}; \pmb{\vartheta}^-) \) with parameters \( \pmb{\vartheta}^- = \pmb{\vartheta} \)

\While{episode $\leq \mathrm{epMax}$}
    \State Initialize M-UAV starting state \( \mathtt{s} = (\q_i, t_{\mathrm{rem}} = N) \)    
        \For{each time step \( n \) within the episode}
            \State Observe current state \( \mathtt{s} \)
            \State Select action $\mathtt{a}$ via \eqref{epsilon-greedy}, perform the action and observe new state
            $\mathtt{s}'=(\q_\mathrm{next}, t_{\mathrm{rem}} \gets t_{\mathrm{rem}}-\delta_t)$
            \State Compute reward \( \mathtt{r} \) for new state $\mathtt{s}'$ using \eqref{reward}
            \State Store transition \( (\mathtt{s}, \mathtt{a}, \mathtt{r}, \mathtt{s}') \) in replay buffer \( \mathcal{D} \)
            
            \State Sample a mini-batch from \( \mathcal{D} \), i.e., \{\( (\mathtt{s}, \mathtt{a}, \mathtt{r}, \mathtt{s}') \)\} 
            \State Update Q-network using Bellman equation  \eqref{loss_function} 
            \If{target network update condition met}
                \State Update target network parameters: \( \pmb{\vartheta}^- \leftarrow \pmb{\vartheta} \)
            \EndIf
            \If{$\text{IsDone}= \|\q_{\text{next}} - \q_f\| \leq r_f \ \textbf{or} \ t_{\text{rem}} = 0$}
                \State \textbf{Break}
            \EndIf
        \EndFor
        
        \State Decay \( \epsilon \) for reduced exploration in future episodes

    \State episode $\gets$ episode + 1,  
 \EndWhile
\end{algorithmic}
\end{algorithm}

\begin{table}[t]
\centering
\caption{\label{table1}Simulation Parameters}
\scalebox{0.9}{%
\begin{tabular}{lc}
\toprule
\textbf{Parameter} & \textbf{Value} \\
\midrule
\multicolumn{2}{c}{\textbf{M-UAV mission parameters}} \\
\midrule
Initial position ($\q_i$) & $[-95, 95, 50]^\top~\mathrm{m}$ \\
Final region ($\q_f, r_f$) & $([ 95, -95, 50]^\top, 30)~\mathrm{m}$\\
Operational altitude ($z_{\min}, z_{\max}$) & $(20, 100)~\mathrm{m}$ \\
Maximum velocity ($v^{\max}_x, v^{\max}_y, v^{\max}_z$) & $(20, 20, 10)~\mathrm{m/s}$ \\
Timeslot duration ($\delta_t$) & $1$ s \\
Minimum safety distance ($d^{\mathrm{min}}$) & $8~\mathrm{m}$ \\
Maximum mission time ($T$) & $100$~$\mathrm{s}$ \\
Timeslot duration ($\delta_t$) & $1$~$\mathrm{s}$\\
\midrule
\multicolumn{2}{c}{\textbf{Flight power consumption parameters} \cite{TatarMamaghani2022}}\\
\midrule
Weight force ($G_0$) & $29.4~\mathrm{N}$ \\
Blade profile power in hovering state ($P_b$) & $79.86~\mathrm{W}$ \\
Induced power in hovering state ($P_i$) & $88.63~\mathrm{W}$ \\
Profile drag coefficient ($C_0$) & $0.0092~\mathrm{W.(m/s)^{-3}}$ \\
Tip speed of the propeller ($U_{\text{tip}}$) & $120~\mathrm{m/s}$ \\
Mean rotor induced velocity in hovering state ($\nu_0$) & $4.03~\mathrm{m/s}$ \\
\midrule
\multicolumn{2}{c}{\textbf{Communication and sensing parameters}  \cite{khuwaja2018survey}} \\
\midrule
Uplink users ($L$) & 5 \\
Downlink users ($K$) & 10 \\
Transmit antennas ($M^{Tx}_h, M^{Tx}_v$) & $(5, 5)$ \\
Receive antennas ($M^{Rx}_h, M^{Rx}_v$) & $(5, 5)$ \\
Operating frequency ($f_c$) & $10~\mathrm{GHz}$ \\
Path-loss exponent ($\alpha$) & $3$ \\
Excessive attenuation factor ($\varsigma$) & $0.1$ \\
Normalized reference channel gain ($\beta_0$) & $110~\mathrm{dB}$ \\
Rician K-factor ($K_{\text{factor}}$) & $10$ \\
Full bandwidth ($B$) & $30~\mathrm{MHz}$ \\
RMS bandwidth ($B_{\text{rms}}$) & $\sqrt{\pi^2/3}B$ \\
Probabilistic channel coefficients ($E_1, E_2$) & ($9.61, 0.16$)\\
\midrule
\multicolumn{2}{c}{\textbf{Thresholds and scales}} \\
\midrule
UCSR threshold ($\rho^{\ul}$) & $0.1$ \\
DCSR threshold ($\rho^{\dl}$) & $0.5$ \\
CRLB estimation threshold ($\rho^{\est}$) & $0.001$ \\
Convergence criterion for Algorithm \ref{algo1} ($\varepsilon$) & 0.001\\
Reward scaling factors ($\eta_1, \eta_2$) & $(10, 2)$\\
Goal reward ($\mathtt{r}_\mathrm{goal}$)  & $100$\\
Failure penalty ($\mathtt{p}_\mathrm{fail}$) & $-10$\\
\midrule
\multicolumn{2}{c}{\textbf{MDP environment parameters}} \\
\midrule
Lower limit of grid & $[-100,-100, 20]^\top$ m \\
Upper limit of grid & $[100,100, 100]^\top$ m  \\
Speed level set \( \mathtt{vel=\{v_{low}, v_{med}, v_{high}\} }\) & \{$0, v^{max}/2, v^{max}$\}\\
\midrule
\multicolumn{2}{c}{\textbf{DDQN agent options}} \\
\midrule
Size of hidden layers ($n_{h1}, n_{h2}$) & (16, 8)\\
Target smooth factor & $0.01$ \\
Discount factor ($\gamma$) & $0.99$ \\
Mini-batch size ($L_\mathrm{bs}$) & $64$ \\
Initial exploration rate & $1$ \\
Epsilon decay rate & $0.001$ \\
Minimum exploration rate & $0.01$ \\
Target update frequency ($f^{\mathtt{update}}$) & $5$ \\
Number of training episodes ($\mathrm{epMax}$) & 1000\\
Critic optimizer learning rate & $0.001$ \\
Critic optimizer gradient threshold & $1$\\
\bottomrule
\end{tabular}}
\end{table}

{\subsection{Convergence and Complexity Analysis}
\label{sec:conv-complex}

\subsubsection{Convergence of Algorithms}

At each timeslot \(n\), Algorithm~\ref{algo1} solves the convexified follower subproblem \(\Prob{1.4}\) using a successive convex approximation (SCA) strategy. Let  
\(\mathcal{Z} \triangleq \{ \mathbf{V}, \mathbf{W}, \mathbf{p}, \pmb{\mu}, \pmb{\omega}, \mathbf{t}, \mathbf{s}, \pmb{\iota} \}\)  
denote the aggregate parameter vector, and let \(U_1^{(i)}\) represent the objective value at the \(i^{\text{th}}\) SCA iteration. During each iteration, non-convex constraints in the original problem \(\Prob{1.1}\) are approximated using first-order Taylor expansions, evaluated at the local point \(\mathcal{Z}^{(i)}\). These surrogates are tight at the expansion point, ensuring that \(\mathcal{Z}^{(i)}\) remains feasible in the next iteration. This guarantees that the feasible set of the convexified problem at iteration \(i{+}1\) is non-empty, preserving the iterations. Since the SCA approximations provide a lower bound for convex functions (or an upper bound for concave functions), the feasible set of the approximated problem is a subset of the original one. In addition, each convex surrogate is solved optimally and all transmit powers are strictly positive to satisfy constraints such as \eqref{ul_cst}, \eqref{dl_cst}, and \eqref{er_cst}, the sequence of objective values \(\{ U_1^{(i)} \}\) is monotonically non-increasing, i.e., $U_1^{(i)} \leq U_1^{(i-1)}$, and bounded below by zero. Hence, the sequence converges. As a result, the SCA algorithm converges to a KKT point of the original non-convex problem  $\Prob{1.3}$, ensuring a locally optimal solution for Player 1's resource allocation. 

On the other hand, since the leader (M-UAV) updates its policy only after receiving the corresponding follower response, the interaction between the two agents adheres to a \textit{two-time-scale} learning framework. In this setup, the follower optimization is executed to local convergence via SCA (i.e., Algorithm~\ref{algo1}) at each timeslot, while the DDQN-based leader explores and learns on a slower time scale across episodes. Under standard conditions for function approximation and experience replay, and given sufficient exploration, Algorithm~\ref{algo2} converges almost surely to a Stackelberg equilibrium \cite{han2019game}, as experimentally shown in the numerical results.

\subsubsection{Complexity of Follower's Algorithm~\ref{algo1}}

In every SCA round the follower solves an SDP comprising $N_{\mathrm{var}}=(K+1)M_t^{2}+5L+2K+2$ real decision variables: the first term collects the $(K+1)$ complex Hermitian beamforming matrices ($M_t^{2}$ real degrees of freedom each), while the latter term corresponds to the scalar slack variables.
The constraint set contains $6L+3K+3$ second–order/affine inequality constraints plus $N_{\mathrm{LMI}}=K+1$ linear matrix inequality (LMI) constraints. For an interior–point solver \cite{Boyd2006}, which is typically used to solve convex conic problems, each $M_t\times M_t$ LMI contributes $M_t$ to the size of logarithmic–barrier parameter vector, so the total barrier size is $N_{\mathrm{LMI}} + (K+1)M_t$. Let $\varepsilon$ be the convergence accuracy of SCA in Algorithm \ref{algo1} for the $N$ subproblems. Then, the overall follower-side complexity is given by
\begin{align}
\mathrm{C}_{\text{SCA}} =\mathcal{O}&\Bigl(N\,\bigl[(K+1)M_t^{2}+5L+2K+2\bigr]^{2}\nonumber\\
&\times\Bigl[(K+1)M_t+6L+3K+3\Bigr]^{\frac{3}{2}}
\log\frac{1}{\varepsilon}\Bigr).
\end{align}

\subsubsection{Complexity of Leader's Algorithm~\ref{algo2}}

The DDQN algorithm optimizes the M-UAV’s trajectory by training two neural networks over multiple episodes. The complexity is broken down into per-step and overall training costs. Let $n_i$, $n_{h1}$, $n_{h2}$, and $n_o$ denote the number of neurons in the input, first hidden, second hidden, and output layers of the deep Q-network, respectively. The total number of trainable parameters is given by $|\pmb{\vartheta}| = n_i n_{h1} + n_{h1} n_{h2} + n_{h2} n_o$, where $n_i=|\mathcal{S}|$ and $n_o=|\mathcal{A}|$. Let $L_\mathrm{bs}$ represent the mini-batch size, $\mathrm{epMax}$ the total number of training episodes, and $N$ the maximum number of time steps per episode. Each forward pass during action selection and each backpropagation step for stochastic gradient descent (SGD) update incur a computational cost of $\mathcal{O}(|\pmb{\vartheta}|)$ \cite{Sutton2018}. Consequently, the total training complexity over all episodes is

\begin{equation}
\label{eq:DDQN-complexity}
\mathrm{C}_{\text{DDQN}} = \mathcal{O}(\mathrm{epMax} \cdot N \cdot L_\mathrm{bs} \cdot |\pmb{\vartheta}|).
\end{equation}

In terms of memory, the dominant component is the experience replay buffer. Each stored transition includes a state-action pair and the corresponding reward and next-state tuple. Thus, for a buffer of size $|\mathcal{D}|$, the overall memory complexity is $\mathcal{O}(|\mathcal{D}| (n_i + n_o + 2))$. 

\begin{remark} For energy-constrained UAVs, it is preferable to offload the training process of Algorithm~\ref{algo2} to a ground control server or edge node, especially during the offline policy learning phases. During flight, only the forward pass of the trained DDQN is executed for trajectory decisions, which incurs negligible computational overhead and energy cost compared to backpropagation or full-scale training. This separation aligns well with real-time deployment scenarios where onboard computation and battery capacity are limited.
\end{remark}
}

\begin{remark}
     While one may consider a fully model-free DRL framework to solve the Stackelberg game by jointly learning the trajectories and resource allocations, this approach suffers from severe practical limitations. The high-dimensional mixed-integer state-action space renders convergence slow and unstable due to the curse of dimensionality, especially under strict SINR, leakage, and secrecy constraints. In contrast, the proposed SCA-DRL design efficiently exploits the problem structure: SCA efficiently solves the follower's resource allocation, while DRL handles the dynamic nature of the leader’s energy-efficient trajectory optimization. This decomposition not only improves convergence and complexity but also aligns with the sequential nature of Stackelberg dynamics.
\end{remark}

\section{Numerical Results} \label{sec4_simulations}
In this section, we present simulation results to evaluate the performance of the proposed approach for secure ISAC operation with an M-UAV under various configurations. We adopt the simulation parameters listed in Table \ref{table1}.

\begin{figure}[t]
 \centering
\hspace*{-7mm}\subfloat[\label{fig:sub1}]{ \includegraphics[width=0.6\linewidth]{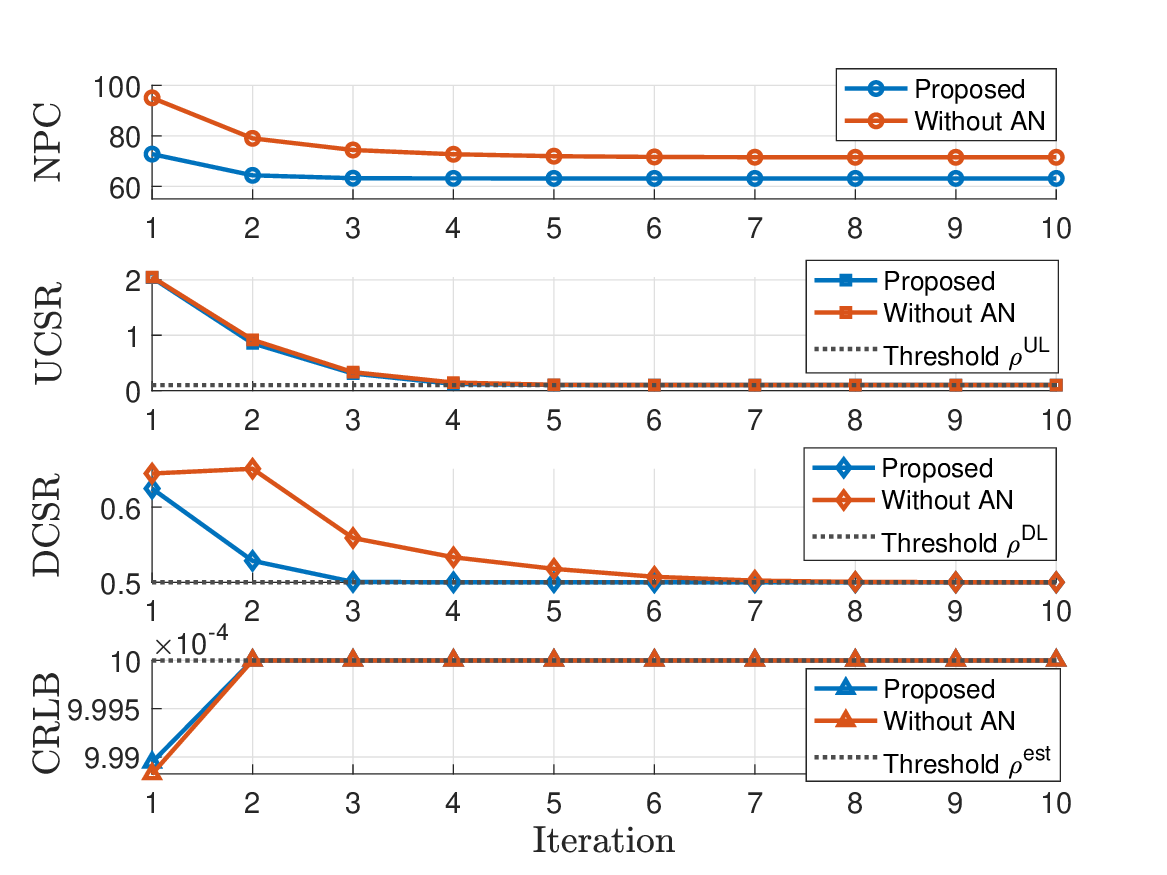}}
~\hspace*{-7mm}
\subfloat[\label{fig:sub2}]{ \includegraphics[width=0.6\linewidth]{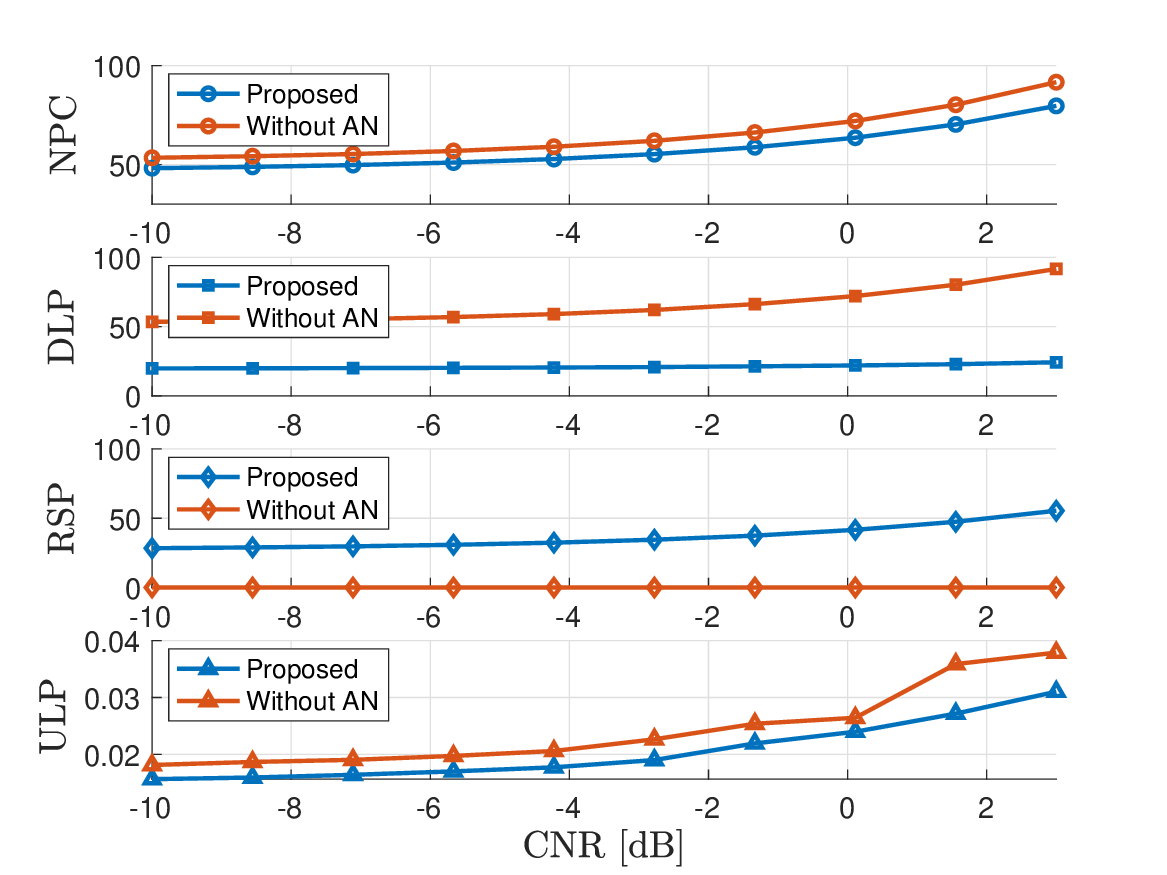}}
  \caption{(a) Algorithm~\ref{algo1}'s convergence versus iteration index, (b) impact of environmental clutter on power allocation for arbitrary channel realizations when $\q_e = [25, 45, 50]^\top$.}
    \label{fig:combined}
\end{figure}

\figref{fig:combined} illustrates the convergence behavior and power efficiency of the proposed low-complexity secure ISAC design (i.e., Algorithm~\ref{algo1}) for Player 1's resource allocation strategy. We also compare our proposed design with a benchmark, labeled as \emph{Without AN}, where the R-BS enforces $\mathbf{W} = \mathbf 0$ in Algorithm~\ref{algo1}, i.e., no dedicated sensing/jamming beamforming. The curves in \figref{fig:sub1} show a rapid decline in the network power consumption (NPC) within the first few iterations, indicating that the algorithms efficiently find solutions with minimal computational overhead, which is critical for dynamic environments. The results also show that our AN-assisted scheme achieves approximately $13\%$ lower NPC compared to the benchmark,  highlighting the dual role of AN in improving power efficiency performance and degrading the adversary's reception quality. In addition, the plots indicate that the algorithm maintains feasible solutions, i.e., guaranteeing compliance with the UCSR, DCSR, and CRLB requirements, which are imposed by the constraints~\eqref{ul_cst}, \eqref{dl_cst}, and \eqref{er_cst},  while gradually refining power allocation and beamforming strategies to minimize the NPC, confirming the accuracy of our analysis. 

In \figref{fig:sub2}, we examine the impact of uncorrelated clutter on the performance of the proposed secure ISAC system. To quantify this effect, the clutter-to-noise ratio (CNR) is defined as the ratio of the clutter power at the R-BS to the antenna noise power at the receiver, assuming that the antenna noise level is identical across all receivers. A CNR below the noise floor indicates a weak clutter environment, where standard signal processing techniques such as Doppler filtering or spatial nulling can effectively mitigate clutter-induced interference. The results demonstrate that as the CNR increases, the proposed design adaptively reallocates more radar sensing power (RSP) to satisfy the CRLB requirement, while the downlink transmission power (DLP) remains relatively unchanged.  On the other hand, the Without-AN scheme requires larger uplink transmission power (ULP) and DLP to preserve the desired communication secrecy and sensing accuracy, leading to substantially greater NPC. This performance gap becomes increasingly pronounced in clutter-rich environments, underscoring the robustness and power efficiency of the proposed AN-assisted scheme in such scenarios.


\begin{figure}[t]
    \centering
\hspace*{-3mm}\subfloat[\label{sim5:sub1}]{ \includegraphics[width=0.5\linewidth]{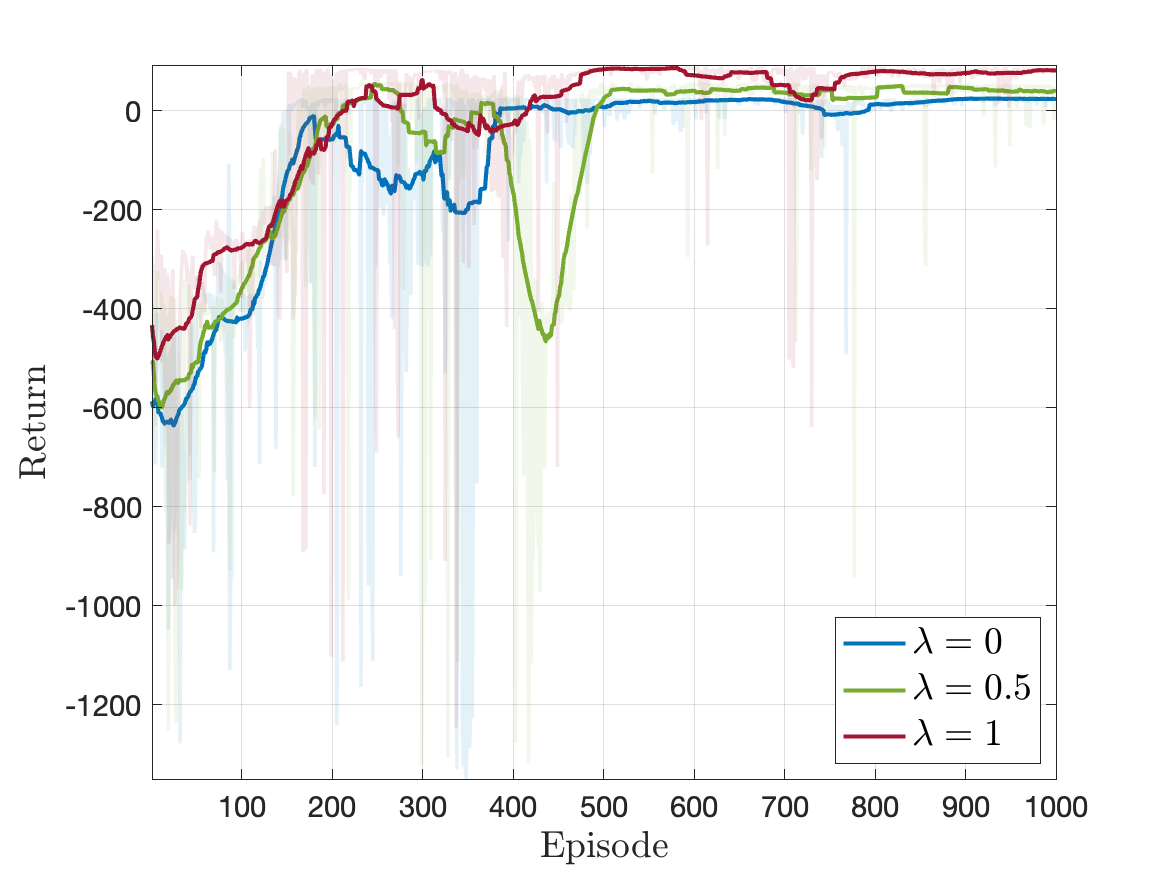}}
~\hspace*{-6mm}
\subfloat[\label{sim5:sub2}]{ \includegraphics[width=0.5\linewidth]{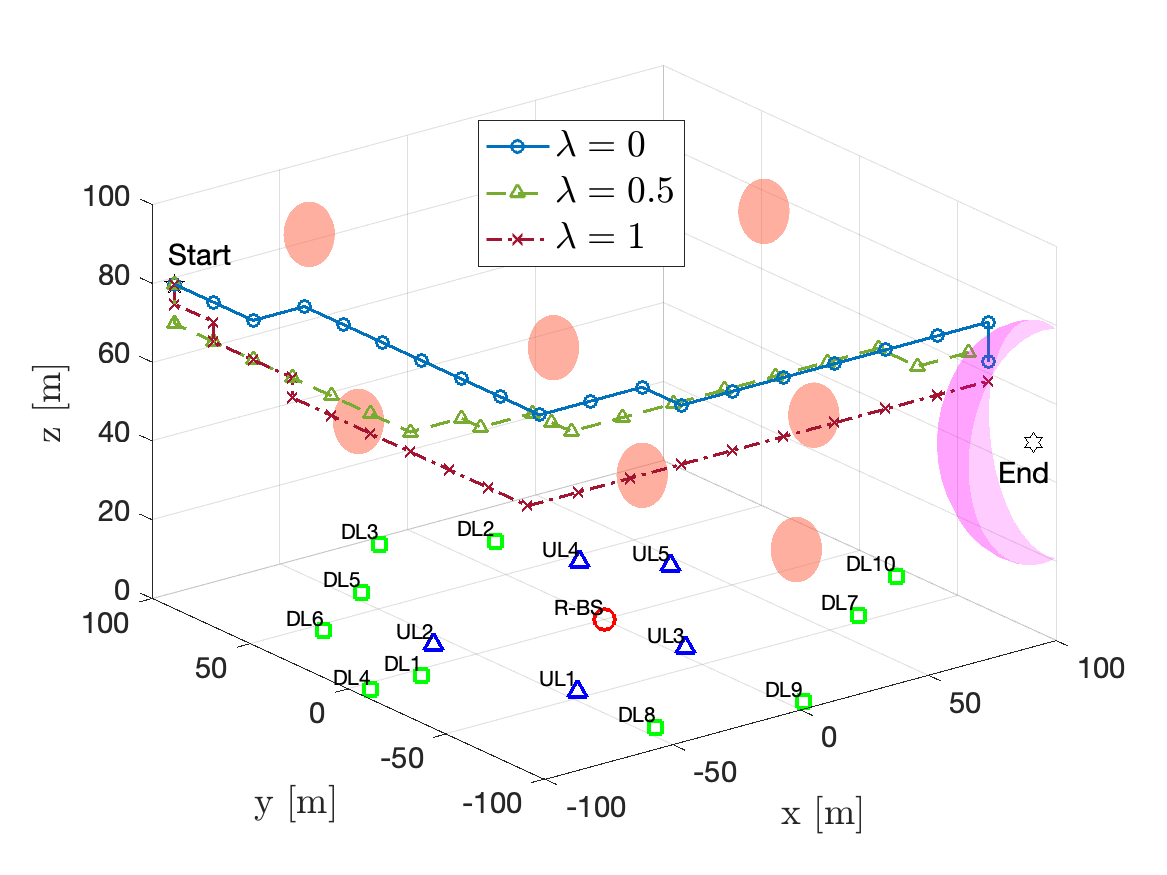}}

\subfloat[\label{sim5:sub3}]{ \includegraphics[width=0.9\linewidth]{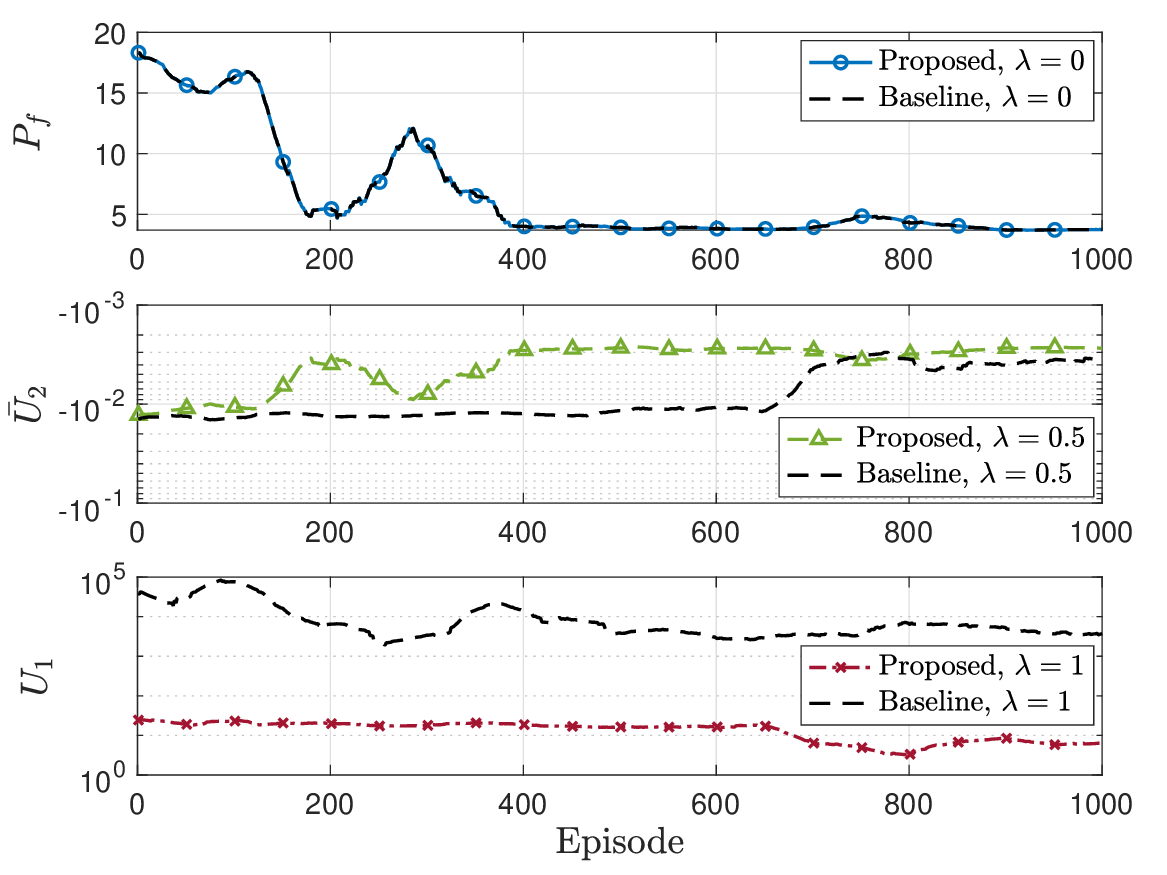}}
    \caption{(a) M-UAV's return across episodes, (b) designed trajectory profiles, (c) Player 2's utility via Algorithm \ref{algo2} for different $\lambda$.}
    \label{sim5:combined}
\end{figure}

\figref{sim5:combined} shows the training performance, M-UAV trajectories, and interactions among players in the  NSG framework designed using Algorithm \ref{algo2}, for different scaling factors $\lambda$. \figref{sim5:sub1} highlights the convergence of the discounted cumulative reward (return), reflecting the M-UAV's ability to gradually learn and optimize its trajectory by balancing its own power efficiency and the legitimate network's power allocation. The convergence trend shows that for all $\lambda$ values, the return stabilizes after approximately $500$ episodes. For $\lambda =0$,  the M-UAV prioritizes its own power reduction over that of Player 1, while $\lambda=1$
achieves the highest return reflecting the M-UAV's focus on maximizing Player 1's power consumption.

In \figref{sim5:sub2}, the optimized M-UAV trajectories are visualized for different values of $\lambda$, representing the trade-off between the M-UAV’s power efficiency and Player 1's power maximization. In all scenarios, the M-UAV successfully avoids collisions with obstacles (red balls) and finds an efficient path to the destination. For $\lambda=0$, the M-UAV prioritizes its power efficiency, resulting in a conservative trajectory with $N=20$ steps and minimal altitude changes to reduce energy consumption. Conversely, $\lambda=1$ causes the M-UAV to adjust its trajectory to maximize Player 1's power consumption, requiring $N=22$ steps as it moves farther from the R-BS, thereby increasing demand on transmit power to satisfy the sensing and secrecy requirements. For intermediate values $\lambda=0.5$, the M-UAV adopts a balanced path-planning strategy. 

\figref{sim5:sub3} illustrates Player 2’s utility over the episodes, demonstrating the interactions between the network and the M-UAV via the proposed NSG solution. The curves labeled \textit{Baseline} correspond to a scenario where Player 1 does not implement any measures to thwart the M-UAV, using only the initial transmit powers and beamforming matrices from Algorithm \ref{algo1}. In contrast, the \textit{Proposed} curves represent the network design based on our approach. For $\lambda=0$, the M-UAV minimizes its power consumption, completely disregarding the countermeasures of the legitimate network. As $\lambda$ increases, the M-UAV adapts its strategy to increase the power consumption of the legitimate network at the cost of increasing its own power usage. The proposed algorithm clearly outperforms the baseline scheme in terms of $\bar{U}_2$. Player 2's regularized utility function normalizes the power consumption of the network and M-UAV by their respective maximum values for fairness. When $\lambda=1$, the NSG reduces to a traditional zero-sum game, where Player 1 minimizes $U_1$ while Player 2 proactively increases it, resulting in an equilibrium. The above results validate the effectiveness of the proposed DRL approach in the Stackelberg game framework. 

\section{Conclusion}\label{sec5_conclusions}
This paper has presented a game-theoretic framework for securing an FD-ISAC system against a proactive aerial adversary. By modeling the interaction between the legitimate network and the M-UAV as an NSG, we jointly optimized the legitimate network's power allocation as well as the M-UAV's power usage and trajectory planning over time, ensuring an energy-efficient adaptive strategy for both players. The proposed solution combines SCA for network resource management with DRL to dynamically adjust the M-UAV’s trajectory. Our simulation results show that the solution efficiently achieves a Stackelberg equilibrium, balancing network power efficiency, communication secrecy, and sensing accuracy.  
Simulation results demonstrate that our proposed solution effectively mitigates the interception capabilities of the M-UAV, maintaining robust DL and UL secrecy rates while ensuring required sensing performance. The convergence behavior of our algorithms confirmed the efficiency of the solution in dynamic environments. 
Our study not only demonstrates the potential of using game theory empowered by convex optimization and DRL in secure ISAC systems, but also paves the way for addressing dynamic threats in evolving 6G networks via novel hybrid optimization techniques. \textcolor{black}{Future work will extend our SCA‐DRL framework to a more realistic continuous state–action framework under stochastic and partially observable environment dynamics, while considering more adaptive reward shaping via generative learning models and adversarial strategies to generalize the approach to more diverse ISAC scenarios.}

\appendices
\numberwithin{equation}{section}
\makeatletter 
\newcommand{\section@cntformat}{Appendix \thesection:\ }
\makeatother

\section{Proof of Theorem \ref{theorem1}}\label{Appendix_A}

Rewrite the signal vector \eqref{yA_tilde} as 
\begin{align}\label{y_observed}
    \y(t) = \H(t) \x(t-\tau) + \n(t),
\end{align}
where $\H(t) = \zeta_0(t)\A_0(t)$. Recall that the CRLB gives the minimum variance attainable by any unbiased estimate \( \hat{\tau} \), and is defined by \cite{kay1993fundamentals}
\begin{align}\label{crlb_theorem}
\Var(\hat{\tau}) \geq \frac{1}{\mathbf{J}(\tau)}=\crlb(\tau),
\end{align}
where $\mathbf{J}(\tau)$ is the Fisher Information matrix (FIM) 
\begin{align}\label{FI_definition}
\hspace{-3mm}\mathbf{J}(\tau)  = \E \left( \left[ \frac{\partial \ln \Lambda(\y; \tau)}{\partial \tau} \right]^2 \right) = -\E \left( \frac{\partial^2 \ln \Lambda(\y; \tau)}{\partial \tau^2}\right),
\end{align}
and $\Lambda(\y; \tau)$ is the corresponding likelihood function. We derive the CRLB for delay estimation by formulating the likelihood function for the observed signal \eqref{y_observed}. Since the noise is Gaussian, the likelihood function is expressed as \cite{bliss2013adaptive}
\begin{align}
    \Lambda(\y; \tau) = \frac{\exp\left( - (\y- \H \x)^\dagger \boldsymbol{\Sigma}_n^{-1} (\y- \H \x) \right)}{\pi^{N_r} \det(\boldsymbol{\Sigma}_n)}.
    \end{align}    
We calculate the FIM by taking the second derivative of the log-likelihood function w.r.t. $\tau$. Defining the residual vector as \( \d \triangleq \y(t) -\H\x(t - \tau) \), the log-likelihood function becomes
\begin{align}
    \ln \Lambda(\y; \tau) = - N_r \ln(\pi) - \ln(\det(\boldsymbol{\Sigma}_n)) - \d^\dagger \boldsymbol{\Sigma}_n^{-1} \d.
\end{align}
Taking the derivative with respect to  $\tau$, we obtain
\begin{align}\label{log_likehood_ode}
 \hspace{-3mm}-\frac{\partial \ln \Lambda(\y; \tau)}{\partial \tau} =  \frac{\partial \d^\dagger \boldsymbol{\Sigma}_n^{-1} \d}{\partial \tau} = 2 \Re \left\{ \left( \frac{\partial \d}{\partial \tau} \right)^\dagger \boldsymbol{\Sigma}_n^{-1} \d \right\}.
\end{align}
In addition, denoting $\dot{\x}(t-\tau)\triangleq \frac{\partial \x(t-\tau)}{\partial \tau}$, we have
\begin{align}
\frac{\partial \d}{\partial \tau} = \H \dot{\x}(t - \tau).
\end{align}
Accordingly, we can simplify \eqref{log_likehood_ode} as
\begin{align}
\frac{\partial \ln \Lambda(\y; \tau)}{\partial \tau} = -2 \Re \left\{ \left[\H \dot{\x}(t - \tau)\right]^\dagger \boldsymbol{\Sigma}_n^{-1} \d \right\}.
\end{align}
The second derivative of the log-likelihood is
\begin{align}\label{sec_der_loglike}
\hspace{-4mm}\frac{\partial^2 \ln \Lambda(\y; \tau)}{\partial \tau^2} \approx -2 \Re \left\{ \left[\H \dot{\x}(t - \tau)\right]^\dagger \boldsymbol{\Sigma}_n^{-1} \H \dot{\x}(t - \tau) \right\}.
\end{align}
The above approximation follows by ignoring the second derivative of $\x$ since it is negligible for typical signal models (e.g., constant-envelope or band-limited signals). Therefore, substituting \eqref{sec_der_loglike} into \eqref{FI_definition} and using the properties of the trace function, the FIM is calculated as
\begin{align}\label{FIM_calc}
\mathbf{J}(\tau) &=  -\E \left( \frac{\partial^2 \ln \Lambda(\y; \tau)}{\partial \tau^2}\right) \nonumber\\
&=  2\tr \left( \H\E\left(\dot{\x}(t - \tau)^\dagger  \dot{\x}(t - \tau)\right) \H^\dagger  \boldsymbol{\Sigma}_n^{-1} \right).
\end{align}

Note that $\E \left( \dot{\x}^\dagger(t - \tau) \dot{\x}(t - \tau) \right)
$ represents the energy of the time derivative of the signal $\x$, which is related to the signal bandwidth. In general, the energy of a bandlimited signal can be approximated using Parseval's theorem as \cite{bliss2013adaptive} 
\begin{align}\label{signaldot_approx}
    \E \left( \dot{\x}^\dagger(t - \tau) \dot{\x}(t - \tau) \right) \approx 4 \pi^2 B_{rms}^2 \S.
\end{align} 
Substituting \eqref{signaldot_approx} into \eqref{FIM_calc}  and using  \eqref{crlb_theorem} results in the CRLB approximation for the estimate of the delay $\tau$ given in Theorem \ref{theorem1}, which completes the proof.
\hfill$\blacksquare$

\section{Proof of Rank-One Optimality}\label{app:rank1}
We prove that, at every iteration of the SCA algorithm, the SDR adopted in solving Problem~$\Prob{1.4}$ through Algorithm~\ref{algo1} yields rank-one optimal solutions for the downlink covariance matrices $\{\V_k^\star\}_{k\in\K}$, i.e., $\rank(\V_k^\star)=1$ for all $k\in\K$. Thus, each beamforming vector $\mathbf v_k$ can be recovered directly from the dominant eigenpair of $\V_k^\star$, without any need for Gaussian randomization or rank-reduction methods, using the singular value decomposition (SVD) as follows:
\[
\v_k = \sqrt{\lambda_k} \mathbf{u}_k, \quad \text{where } \V_k^\star = \lambda_k \mathbf{u}_k \mathbf{u}_k^\dagger.
\]

To this end, since $\Prob{1.4}$ is convex with zero duality gap, Slater’s condition holds, implying that the Karush–Kuhn–Tucker (KKT) conditions are therefore necessary and sufficient for optimality \cite{Boyd2006}.
Let $\Z_k\succeq\mathbf0$ denote the dual matrix associated with the semidefinite constraint $\V_k\succeq\mathbf0$. Focusing only on the subset of constraints in $\Prob{1.4}$ that involve $\V_k$, we express each such constraint in the general form \( g_i(\V_k) \le 0 \). The partial Lagrangian, ignoring terms independent of $\V_k$, can then be represented as
\begin{align}
\L &\!=\!  \tr(\V_k)\!-\!  \lambda_k \tr(\V_k \Ha{k})  \!-\!\beta_k \tr(\boldsymbol{Y}\V_k\boldsymbol{Y}^\dagger) \!+\! \tr(\Z_k \V_k)\nonumber\\  
&\hspace{-5mm}\!+\!(\eta_k\!-\!\bar{\eta}_k) \tr(\V_k \Ha{e})\!+\! \sum_{\ell=1}^L \rho_\ell\, |\zeta_0|^2 \, \tr( \A_0^\dagger \B_\ell \A_0 \V_k),
\end{align}
where $\lambda_k, \eta_k, \bar{\eta}_k, \rho_\ell  \ge 0$ are the nonnegative dual multipliers, and
$\B_\ell \triangleq \left[\Psi_\ell^{\text{lo}}\right]^{-1} \h_{\ell,a} \h_{\ell,a}^\dagger \left[\Psi_\ell^{\text{lo}}\right]^{-1}$. Thus, the dual problem becomes an unconstrained optimization problem in the form $\max_{\V_k} \mathcal{L}$. Setting the derivative of $\L$ w.r.t. $\V_k$ to zero yields
\begin{align}
\frac{\partial \L}{\partial \V_k} 
&= \I 
   - \lambda_k \Ha{k} 
   - \beta_k \boldsymbol{Y}^\dagger \boldsymbol{Y}
   + (\eta_k - \bar{\eta}_k) \Ha{e} \nonumber\\
   &\quad + \sum_{\ell=1}^L \rho_\ell\, |\zeta_0|^2\, \A_0^\dagger \B_\ell \A_0 
   + \Z_k = \mathbf{0},
\end{align}
which implies 
$\Z_k = \lambda_k \Ha{k} + \beta_k \boldsymbol{Y}^\dagger \boldsymbol{Y}  + \bar{\eta}_k \Ha{e}- \M_k$, where $\M_k \triangleq \I + \eta_k \Ha{e} + \sum_{\ell=1}^L \rho_\ell\, |\zeta_0|^2\, \A_0^\dagger \B_\ell \A_0$.

Note that the first term in $\M_k$ is strictly positive definite, i.e., $\I \succ \mathbf{0}$. The second term $\eta_k \Ha{e} \succeq \mathbf{0}$ is PSD since it represents a nonnegative scaling of the rank-one Hermitian matrix $\Ha{e} = \h_{a,e} \h_{a,e}^\dagger$. Finally, each term $\A_0^\dagger \B_\ell \A_0$ in the summation is also PSD since $\B_\ell \succeq \mathbf{0}$ and congruence with the full-rank matrix $\A_0$ preserves positive semidefiniteness. Thus, all terms in $\M_k$ are Hermitian and PSD, and the presence of $\I \succ \mathbf{0}$ ensures that $\M_k \succ \mathbf{0}$. Likewise, we have $\pmb{\Phi}\triangleq \lambda_k \Ha{k} + \beta_k \boldsymbol{Y}^\dagger \boldsymbol{Y}  + \bar{\eta}_k \Ha{e} \succ 0$.

By complementary slackness, the optimal solution satisfies $\Z_k \V_k^\star = \mathbf{0}$. Substituting in the expression for $\Z_k$, we obtain
\begin{align}\label{eq:eig_relation}
\pmb{\Phi} \V_k^\star = \M_k \V_k^\star.
\end{align}
Since $\M_k \succ \mathbf{0}$, it admits a unique Hermitian square root $\M_k^{1/2} \succ \mathbf{0}$ \cite{horn2012matrix}. We apply the congruence transformation
\begin{align}\label{transform_congru}
\widetilde{\V}_k \triangleq \M_k^{1/2} \V_k^\star \M_k^{1/2},
\end{align}
and substitute $\V_k^\star = \M_k^{-1/2} \widetilde{\V}_k \M_k^{-1/2}$ into \eqref{eq:eig_relation}. After simplification, we have the eigenvalue relation
\begin{align}\label{eq:rank1_relation}
\widetilde{\mathbf{H}}_k \widetilde{\V}_k = \widetilde{\V}_k,
\end{align}
where
\[
\widetilde{\mathbf{H}}_k \triangleq \lambda_k \widetilde{\mathbf{h}}_k \widetilde{\mathbf{h}}_k^\dagger + \beta_k \widetilde{\boldsymbol{y}} \widetilde{\boldsymbol{y}}^\dagger + \bar{\eta}_k \widetilde{\mathbf{h}}_e \widetilde{\mathbf{h}}_e^\dagger,
\]
and 
$
\widetilde{\mathbf{h}}_k \triangleq \M_k^{-1/2} \mathbf{h}_{a,k},~\widetilde{\boldsymbol{y}} \triangleq \M_k^{-1/2} \boldsymbol{Y}^\dagger,~\widetilde{\mathbf{h}}_e \triangleq \M_k^{-1/2} \mathbf{h}_{a,e}.
$
Each term in $\widetilde{\mathbf{H}}_k$ is a Hermitian rank-one PSD matrix. Consequently, $\widetilde{\mathbf{H}}_k$ is Hermitian and has rank at most three. Since $\widetilde{\mathbf{H}}_k \widetilde{\V}_k = \widetilde{\V}_k$, the range of $\widetilde{\V}_k$ lies in the span of $\{\widetilde{\mathbf{h}}_k, \widetilde{\boldsymbol{y}}, \widetilde{\mathbf{h}}_e\}$. Thus, $\widetilde{\V}_k$ lies in a cone spanned by at most three rank-one PSD matrices. Since the objective function is linear in $\V_k$ (and hence in $\widetilde{\V}_k$), the optimum is attained at an extreme point of the feasible set. The extreme points of the PSD cone under linear constraints are rank-one matrices \cite{Huang2010Rank1, le2023rank}. Hence, $\widetilde{\V}_k$ is rank-one. 

The above result can be obtained using another perspective. Suppose, for contradiction, that the optimal solution $\widetilde{\V}_k$ has rank $r > 1$. Let its eigen-decomposition be expressed as
\begin{align}
\widetilde{\V}_k = \sum_{j=1}^r \sigma_j \mathbf{u}_j \mathbf{u}_j^\dagger, \quad \sigma_j > 0, \; \mathbf{u}_i^\dagger \mathbf{u}_j = \delta_{ij}.
\end{align}
where $\delta_{ij}=\delta(i-j)$ is the Dirac delta function. Since every matrix that appears in the span of $\widetilde{\V}_k$, i.e., 
$\widetilde{\mathbf{h}}_k  \widetilde{\mathbf{h}}_k^\dagger$,
$\widetilde{\boldsymbol{y}}\widetilde{\boldsymbol{y}}^\dagger$, and
$\widetilde{\mathbf{h}}_e \widetilde{\mathbf{h}}_e^\dagger$, is itself rank one, each such constraint depends on $\widetilde{\V}_k$ only through
the \emph{scalar} inner products
$\mathbf{u}_j^\dagger \widetilde{\mathbf{q}}\,
 \widetilde{\mathbf{q}}^\dagger \mathbf{u}_j
 = |\widetilde{\mathbf{q}}^\dagger\mathbf{u}_j|^2$
with $\widetilde{\mathbf{q}}\in\bigl\{\widetilde{\mathbf{h}}_k,
\widetilde{\boldsymbol{y}},\widetilde{\mathbf{h}}_e\bigr\}$.
Consequently, if we allocate the \emph{entire} power
$ \sum_{j=1}^r \sigma_j$ to whichever eigen-vector
$\mathbf{u}_j$ yields the largest value of
$|\widetilde{\mathbf{h}}_k^\dagger\mathbf{u}_j|^2$
(the eigenvector that couples most strongly to the intended user‐channel), all linear constraints involving $\widetilde{\V}_k$ remain satisfied or can be satisfied with at most the same total power. 

In light of this, define the index of the best-aligned component with the intended channel as
$j^\star \!=\! \operatorname{argmax}_{1 \le j \le r} |\widetilde{\mathbf{h}}_k^\dagger \mathbf{u}_j|^2$, and construct a rank-one matrix
\begin{align}
\widehat{\V}_k = \left( \sum_{j=1}^r \sigma_j \right) \mathbf{u}_{j^\star} \mathbf{u}_{j^\star}^\dagger.
\end{align}
Leveraging the properties of rank-one Hermitian matrices and orthonormality, the projection of any rank-one constraint matrix (e.g., $\widetilde{\mathbf{h}}_k \widetilde{\mathbf{h}}_k^\dagger$) onto $\widehat{\V}_k$ is always greater than or equal to its projection onto $\widetilde{\V}_k$, i.e.,
$\tr\left( \widetilde{\mathbf{h}}_k \widetilde{\mathbf{h}}_k^\dagger \widehat{\V}_k \right) \ge \tr\left( \widetilde{\mathbf{h}}_k \widetilde{\mathbf{h}}_k^\dagger \widetilde{\V}_k \right)$, and similar inequalities hold for the radar leakage and eavesdropper constraints in the primal problem. To ensure that $\widehat{\V}_k$ satisfies all the constraints with the same margins, we introduce a scaling factor
\begin{align}
t = \min_{\widetilde{\mathbf{q}} \in \{ \widetilde{\mathbf{h}}_k, \widetilde{\boldsymbol{y}}, \widetilde{\mathbf{h}}_e \}} \frac{ \tr(\widetilde{\mathbf{q}} \widetilde{\mathbf{q}}^\dagger \widetilde{\V}_k) }{ \left( \sum_{j=1}^r \sigma_j \right) |\widetilde{\mathbf{q}}^\dagger \mathbf{u}_{j^\star}|^2 },
\end{align}
and define the rescaled feasible point $\overline{\V}_k = t \, \widehat{\V}_k$. Since $t \in (0,1]$, we have
\begin{align}
\tr(\overline{\V}_k) = t \sum_{j=1}^r \sigma_j \leq \sum_{j=1}^r \sigma_j = \tr(\widetilde{\V}_k),
\end{align}
which implies that $\overline{\V}_k$ achieves a lower or equal objective value while remaining feasible, which contradicts the optimality of $\widetilde{\V}_k$. Hence, the original assumption that $r > 1$ must be false, and we conclude that $\widetilde{\V}_k$ must be rank-one.

Since congruence transformation \eqref{transform_congru} preserves rank \cite{horn2012matrix}, we obtain $\rank(\V_k^\star) = \rank(\widetilde{\V}_k) = 1,~k\in\K$. Thus, the SDR in Problem~$\Prob{1.4}$ yields a rank-one optimal solution. \hfill$\blacksquare$

\bibliographystyle{IEEEtran}
\bibliography{references.bib}

\end{document}